\documentclass[12pt]{elsarticle}
\usepackage[margin=2.5cm]{geometry}
\usepackage{lipsum}

\makeatletter
\def\ps@pprintTitle{%
   \let\@oddhead\@empty
   \let\@evenhead\@empty
   \def\@oddfoot{\reset@font\hfil\thepage\hfil}
   \let\@evenfoot\@oddfoot
}
\makeatother

\usepackage{amsfonts,color,morefloats,pslatex}
\usepackage{amssymb,amsthm, amsmath,latexsym}

\newtheorem{theorem}{Theorem}
\newtheorem{lemma}[theorem]{Lemma}

\newcommand{\rank}{{\mathrm{rank}}}

\newcommand{\gf}{{\mathbb{F}}}
\newcommand{\PG}{{\mathrm{PG}}}

\newcommand{\support}{{\mathrm{suppt}}}

\newcommand{\C}{{\mathcal{C}}}

\newcommand{\cH}{{\mathcal{H}}}

\newcommand{\bc}{{\mathbf{c}}}
\newcommand{\bg}{{\mathbf{g}}}

\usepackage{blindtext}

\begin{document}

\begin{frontmatter}



\title{Constructions of near MDS codes which are optimal locally recoverable codes}

\tnotetext[fn1]{
Z. Heng's research was supported in part by the National Natural Science Foundation of China under Grant 11901049, in part by the  Young Talent Fund of University Association for Science and Technology in Shaanxi, China, under Grant 20200505 and in part by the Fundamental Research Funds for the Central Universities, CHD, under Grant 300102122202.
}

\author{Xiaoru Li}
\ead{lixiaoru@163.com}
\author{Ziling Heng \corref{cor}}
\ead{zilingheng@chd.edu.cn}

\cortext[cor]{Corresponding author}
\address{School of Science, Chang'an University, Xi'an 710064, China}




\begin{abstract}
A linear code with parameters $[n,k,n-k]$ is said to be  almost maximum distance separable (AMDS for short). An AMDS code whose dual is also AMDS is referred to as an near
maximum distance separable (NMDS for short) code. NMDS codes have nice applications in finite geometry, combinatorics, cryptography and data storage. In this paper, we first present several constructions of NMDS codes and determine their weight enumerators. In particular, some constructions produce NMDS codes with the same parameters but different weight enumerators. Then we determine the locality of the NMDS codes and obtain many families of distance-optimal and dimension-optimal locally repairable codes.
\end{abstract}


\end{frontmatter}


\section{Introduction}
Let $q$ be a prime power and  $\mathbb{F}_{q}$ the finite field with $q$ elements. Denote by $\mathbb{F}_{q}^*=\mathbb{F}_{q}\setminus \{0\}$.

Linear codes are an important class of error-correcting codes which are widely used in communication systems. For a non-empty set $\mathcal{C} \subseteq \mathbb{F}_{q}^{n}$, if $\mathcal{C}$ is a $k$-dimensional linear subspace over $\mathbb{F}_{q}$, then it is called an $[n,k,d]$ linear code over $\mathbb{F}_q$, where $d$ denotes its minimal  distance. The dual code of an $[n,k]$ linear code $\mathcal{C}$ over  $\mathbb{F}_{q}$ is defined as
$$
\mathcal{C}^{\perp}=\left\{ \mathbf{c}^{\perp} \in \mathbb{F}_{q}^{n}: \langle  \mathbf{c}^{\perp}, \mathbf{c} \rangle=0\mbox{ $\forall$ }\mathbf{c} \in \mathcal{C} \right\},
$$
where $\langle \mathbf{c}^{\perp},\mathbf{c} \rangle$ denotes the Euclidean inner product of $\mathbf{c}^{\perp}$ and $\mathbf{c}$. Then $\mathcal{C}^{\perp}$ is an $[n,n-k]$ linear code. Let $A_{i}$ denote the number of codewords with weight $i$ in a linear code $\mathcal{C}$ of length $n$, where $0 \leq i \leq n$. The sequence $(1,A_{1},A_{2}, \cdots ,A_{n})$ is called the weight distribution of $\mathcal{C}$. The polynomial
$$
A(z)=1+A_{1}z+A_{2}z^2+ \cdots +A_{n}z^n
$$
is called the weight enumerator of $\mathcal{C}$. Weight distribution is an important research subject as it contains crucial information including the capabilities of error detection and correction. The weight distributions of linear codes have been widely investigated in the literature \cite{D, DT, HDZ, L1, L2, TD, Wqiu}.

A linear code is called an MDS (maximum distance separable) code if it has parameters $[n, k, n-k+1]$. The dual of an MDS code is also an MDS code. The weight distribution of a $q$-ary $[n,k]$ MDS code is unique. MDS codes are meaningful for both theory and practice. Recently, twisted Reed-Solomon
codes with one-dimensional hull which are MDS were constructed in \cite{W}.  Linear complementary dual  MDS codes
of non-Reed-Solomon type were constructed in \cite{WHL}.

A linear code is said to be almost maximum distance separable (almost MDS or AMDS for short) if it has parameters $[n, k, n-k]$.  AMDS codes with dimensions $1, n -2, n-1, n$ are said to be trivial ones.  In general, the dual of an
AMDS code may not be AMDS. It is known that
AMDS codes,
$n$-tracks and linear orthogonal arrays of index $q$ are equivalent to each other \cite{B}. Some upper bounds
on the maximum length for which an AMDS code exists were summarized in \cite{B, FaldumWillems97}.

 If both a code and its dual are AMDS codes, then the code is referred to as an near maximum distance separable(near MDS or NMDS for short) code.  The first near MDS code was the $[11, 6, 5]$ ternary Golay code discovered by Golay \cite{Golay49}. A characterization of
NMDS codes with a parity-check matrix was given in \cite{DodLan95}. Recently, some families of NMDS
codes were constructed. In \cite{DT, TD}, Ding and Tang
constructed several infinite families of NMDS codes holding $t$-designs for $t=2,3,4$. In \cite{J}, self-dual NMDS codes were derived from elliptic curves.
In \cite{HY}, MDS or NMDS self-dual codes from twisted generalized
Reed-Solomon codes were constructed.
In \cite{Wqiu}, several infinite families of NMDS
codes were constructed with oval polynomials.  These families of NMDS codes were proved to be
optimal or almost optimal locally recoverable codes in \cite{TanP}.

According to \cite{GY} and \cite{TanP}, AMDS and NMDS codes can be used to derive optimal or nearly optimal locally recoverable codes. Hence it is interesting to construct more new AMDS or NMDS codes. In \cite{Wqiu}, some special matrixes were used to construct NMDS codes. Inspired by the work in \cite{Wqiu}, we present several new constructions of NMDS codes with different matrixes in this paper. The weight enumerators of the NMDS codes are explicitly determined. In particular, some constructions produce NMDS codes with the same parameters but different weight enumerators. Then we determine the locality of the NMDS codes and obtain many families of distance-optimal and dimension-optimal locally repairable codes.

The rest of this paper is organized as follows. In Section \ref{sec2}, we present some properties of NMDS codes and oval polynomials used in this paper. In Section \ref{sec3}, we give several constructions of NMDS codes and determine their parameters and weight distributions. In Section \ref{sec4}, we prove the NMDS codes constructed in this paper are optimal or almost optimal locally recoverable codes.

\section{Preliminaries}\label{sec2}
In this section, we introduce some properties of NMDS codes and oval polynomials used later.
\subsection{Some properties of NMDS codes}

In this subsection, we present two properties of NMDS codes.

Denote by $(1,A_1,\cdots,A_n)$ and $(1,A_1^\perp,\cdots,A_n^\perp)$ the weight distributions of a linear code $\C$ and its dual $\C^\perp$ of length $n$, respectively. The weight distributions of an NMDS code and its dual satisfy the following recurrence relations.

\begin{lemma}[\cite{DodLan95}]\label{lem-N1}
Let $\mathcal{C}$ be an NMDS code with paraments $[n, k, n-k]$ over the finite field $\gf_q$. Then the weight distributions of the two codes
$\mathcal{C}^\perp$ and $\mathcal{C}$
are given by
\begin{eqnarray*}
A_{k+s}^\perp = \binom{n}{k+s} \sum_{j=0}^{s-1} (-1)^j \binom{k+s}{j}(q^{s-j}-1) +
             (-1)^s \binom{n-k}{s}A_{k}^\perp
\end{eqnarray*}
for $s \in \{1,2, \ldots, n-k\}$; and
\begin{eqnarray*}
A_{n-k+s} = \binom{n}{k-s} \sum_{j=0}^{s-1} (-1)^j \binom{n-k+s}{j}(q^{s-j}-1) +
             (-1)^s \binom{k}{s}A_{n-k}
\end{eqnarray*}
for $s \in \{1,2, \ldots, k\}$.
\end{lemma}

Despite the recurrence relations in Lemma \ref{lem-N1}, the weight distributions of NMDS codes can't be totally determined. Two NMDS codes with the same parameters many have different weight distributions. In this paper, some NMDS codes with the same parameters but different weight distributions will be constructed.

The following is another useful property of NMDS codes.

\begin{lemma}[\cite{FaldumWillems97}]\label{lem-N2}
Let $\mathcal{C}$ be an NMDS code. Then for any minimum weight codeword $\bc$ in $\mathcal{C}$, there exists, up to a multiple, a unique minimum weight codeword $\bc^\perp$ in $\mathcal{C}^\perp$ such that
$\support(\bc) \cap \support(\bc^\perp)=\emptyset$, where $\support(\bc)=\{1 \leq i \leq n: c_i \neq 0\}$
denotes the support of the codeword $\bc=(c_1, \ldots, c_n)$. Particularly, $\mathcal{C}$ and its dual have the same number of minimum weight codewords.
\end{lemma}

\subsection{Some properties of oval polynomials}
In this subsection, let $q=2^m$ with $m$ a positive integer.
In order to calculate the parameters and weight distributions of the NMDS codes in this paper, here we list some properties of oval polynomial. First, we briefly introduce the definition of oval polynomial in the following.

\begin{lemma}[\cite{LN97}]\label{thm-opoly}
Let $m \geq 2$. Any hyperoval in the Desarguesian projective plane $\PG(2, q)$ can be written as
$$
\cH(f)=\{(f(c), c, 1): c \in \gf_q\} \cup \{(1,0,0)\}  \cup \{(0,1,0)\},
$$
where $f \in \gf_q[x]$ is a polynomial such that
\begin{enumerate}
\item $f$ is a permutation polynomial of $\gf_q$ with $\deg(f)<q$ and $f(0)=0$, $f(1)=1$; and
\item for each $a \in \gf(q)$, $g_a(x):=(f(x+a)+f(a))x^{q-2}$ is also a permutation polynomial of $\gf_q$.
\end{enumerate}
Conversely, every such set $\cH(f)$ is a hyperoval.
\end{lemma}

If a polynomial meets the two conditions of Lemma \ref{thm-opoly}, we call it an \emph{oval polynomial}. It is easy to deduce that $f(x)=x^2$ is an oval polynomial over $\gf_q$ for all $m \geq 2$. The followings are some properties of oval polynomials.

\begin{lemma}[\cite{Masch98}]\label{lem-opoly1}
A polynomial $f$ with $f(0)=0$ over $\gf_q$ is an oval polynomial if and only if $f_u:=f(x)+ux$ is $2$-to-$1$ for any $u \in \gf_q^*$.
\end{lemma}

\begin{lemma}[\cite{Masch98}]\label{lem-opoly2}
$f$ is an oval polynomial over $\gf_q$ if and only if
\begin{enumerate}
\item $f$ is a permutation of $\gf_q$; and
\item
$$
\frac{f(x)+f(y)}{x+y} \neq \frac{f(x)+f(z)}{x+z}
$$
for all pairwise-distinct $x, y, z$ in $\gf_q$.
\end{enumerate}
\end{lemma}

\begin{lemma}[\cite{Wqiu}]\label{lem-opoly3}
Let $m \geq 3$ be odd and let $f(x)$ be an oval polynomial over $\gf_q$ with coefficients in $\gf_2$. Then $f(x)+x+1=0$ has no solution in $\gf_q$.
\end{lemma}

\section{Several constructions of NMDS codes}\label{sec3}
In this section, we first present several constructions of NMDS codes over $\gf_q$ with some special matrixes and the oval polynomial $f(x)=x^2$, and then determine their weight enumerators, where $q=2^m$. For convenience, denote by $\dim(\mathcal{C})$ and $d(\mathcal{C})$ the dimension and minimal distance of a linear code $\mathcal{C}$, respectively.
Let $\alpha_0=0,\alpha_1=1,\alpha_2,\cdots,\alpha_{q-1}$ be all elements of $\gf_q$.
\subsection{NMDS code with parameters $[q+4, 3, q+1]$}
Define
\begin{eqnarray*}
G_1=\left[
\begin{array}{lllllllll}
1 & 1 & \cdots & 1 & 1 & 0 & 0 & 1 & 0\\
\alpha_1 & \alpha_2 & \cdots & \alpha_{q-1} & 0 & 0 & 1 & 0 & 1\\
\alpha_1^2 & \alpha_2^2 & \cdots & \alpha_{q-1}^2 & 0 & 1 & 0 & 1 & 1
\end{array}
\right].
\end{eqnarray*}
Obviously, $G_1$ is a $3$ by $q+4$ matrix over $\gf_q$. Let  $\mathcal{C}$ be the linear code over $\gf_q$ with generator matrix $G_1$. Next, we determine the parameters and weight enumerator of $\C$.

\begin{theorem}\label{thm-N4235}
Let $m \geq 3$ be odd. Then the linear code $\mathcal{C}$ is an NMDS code over $\gf_q$ with parameters $[q+4, 3, q+1]$ and weight enumerator
\begin{eqnarray*}
A(z)=1 + (q-1)(q+2) z^{q+1} + \frac{q(q-1)(q+1)}{2} z^{q+2} + \\
(q-1)(q-2) z^{q+3} + \frac{(q-1)(q^2-3q+2)}{2} z^{q+4}.
\end{eqnarray*}
\end{theorem}

\begin{proof}
Firstly, it is easy to deduce that $\dim(\mathcal{C})=3$ as the $q$-th, $q+1$-th and $q+2$-th columns of the generator matrix $G_1$ are linearly independent.

\subsubsection*{We then prove that $\mathcal{C}^\perp$ has parameters $[q+4, q+1, 3]$.}

Obviously, $\dim(\mathcal{C}^\perp)=(q+4)-\dim(\mathcal{C})=q+1$. It is easy to find that  no column of $G_1$ is the zero vector and any two columns of $G_1$ are $\gf_q$-linearly independent. Then the minimum distance $d(\mathcal{C}^\perp) > 2$. We also find the $q$-th, $q+1$-th, $q+3$-th columns of $G_1$ are linearly dependent, which means that $\mathcal{C}^\perp$ has codewords of weight $3$. Then we have the minimum distance $d(\mathcal{C}^\perp)=3$. Now we calculate the total number of codewords of weight $3$ in $\mathcal{C}^\perp$. We need to consider the following cases.

{Case 1.1:} Let $x, y, z$ be three pairwise different elements in $\gf_q$. Consider the submatrix
\begin{eqnarray*}
M_{1,1}=\left[
\begin{array}{lll}
1 & 1 & 1 \\
x & y & z \\
x^2 & y^2 & z^2
\end{array}
\right].
\end{eqnarray*}
We have $|M_{1,1}|=(x+y)(x^2+z^2)+(x+z)(x^2+y^2)$. By Lemma \ref{lem-opoly2}, $|M_{1,1}| \neq 0$. Hence, the rank of $M_{1,1}$ is three. In conclusion, $\mathcal{C}^\perp$ has no codeword of weight $3$ whose nonzero coordinates are at the first $q$ locations.

{Case 1.2:} Let $x, y$ be two different elements in $\gf_q$. Consider the submatrix
\begin{eqnarray*}
M_{1,2}=\left[
\begin{array}{lll}
1 & 1 & 0 \\
x & y & 0 \\
x^2 & y^2 & 1
\end{array}
\right].
\end{eqnarray*}
Then we have $|M_{1,2}|=y+x$. Since $x \neq y$, $|M_{1,2}| \neq 0$. The rank of $M_{1,2}$ is $3$. Hence, $\mathcal{C}^\perp$ has no codeword of weight $3$ whose first two nonzero coordinates are at the first $q$ locations and the rest is at the $q+1$-th location.

{Case 1.3:} Let $x, y$ be two different elements in $\gf_q$. Consider the submatrix
\begin{eqnarray*}
M_{1,3}=\left[
\begin{array}{lll}
1 & 1 & 0 \\
x & y & 1 \\
x^2 & y^2 & 0
\end{array}
\right].
\end{eqnarray*}
Then we have $|M_{1,3}|=y^2+x^2=(x+y)^2$. Since $x \neq y$, $|M_{1,3}| \neq 0$. The rank of $M_{1,3}$ is $3$. Hence, $\mathcal{C}^\perp$ has no codeword of weight $3$ whose first two nonzero coordinates are at the first $q$ locations and the rest is at the $q+2$-th location.

{Case 1.4:} Let $x, y$ be two different elements in $\gf_q$. Consider the submatrix
\begin{eqnarray*}
M_{1,4}=\left[
\begin{array}{lll}
1 & 1 & 1 \\
x & y & 0 \\
x^2 & y^2 & 1
\end{array}
\right].
\end{eqnarray*}
It is easy to deduce that $|M_{1,4}|=(y^2+1)x+(x^2+1)y$. If $x=0, y=1$ or $x=1, y=0$, then $|M_{1,4}| \neq 0$ and $\rank(M_{1,4})=3$. Hence, $\mathcal{C}^\perp$ does not have a codeword of weight $3$ whose coordinates are at the first, $q$-th and $q+3$-th locations. Next we count the number of pair $(x,y)$ such that $|M_{1,4}|=0$, where $x,y \in \gf_q \setminus \{0,1\}$ . For any $x,y \in \gf_q \setminus \{0,1\}$, let $a=\frac{x^2+1}{x}$. Then $a \neq 0$. By Lemma \ref{lem-opoly3}, $a \neq 1$. We have $z^2+az$ is $2$-to-$1$ by Lemma \ref{lem-opoly1}. Therefore, there exists another element $y \in \gf_q \setminus \{0,1\}$ such that $x^2+ax=1=y^2+ay$. For this pair $(x, y)$ we have $|M_{1,4}|=0$ and vice versa. Hence, the number of distinct $(x,y) \in \gf_q \setminus \{0,1\}$ such that $|M_{1,4}|=0$ is equal to $(q-2)/2$. Consequently, the number of codewords of weight $3$ in $\mathcal{C}^\perp$  whose first two nonzero coordinates are at the first $q$ locations (expect the first and $q$-th locations) and the rest is at the $q+3$-th location is equal to $(q-2)(q-1)/2$.

{Case 1.5:} Let $x, y$ be two distinct elements in $\gf_q$. Consider the submatrix
\begin{eqnarray*}
M_{1,5}=\left[
\begin{array}{lll}
1 & 1 & 0 \\
x & y & 1 \\
x^2 & y^2 & 1
\end{array}
\right].
\end{eqnarray*}
It is easy to deduce that $|M_{1,5}|=x^2+y^2+x+y$. Now we calculate the number of $(x,y)$ such that $|M_{1,5}|=0$. Let $|M_{1,5}|=x^2+y^2+x+y=0$ which is equivalent to $(x+y)^2=x+y$. Since $x \neq y$, then this equation  can be simplified to $x+y=1$. We then deduce that the number of different $(x,y)$ such that $|M_{1,5}|=0$ is equal to $q/2$. In conclusion, the number of codewords of weight $3$ in $\mathcal{C}^\perp$ whose first two nonzero coordinates are at the first $q$ locations and the rest is at the $q+4$-th location is equal to $q(q-1)/2$.

{Case 1.6:} Let $x$ be an element in $\gf_q$. Consider the submatrix
\begin{eqnarray*}
M_{1,6}=\left[
\begin{array}{lll}
1 & 0 & 0 \\
x & 0 & 1 \\
x^2 & 1 & 0
\end{array}
\right].
\end{eqnarray*}
Then we have $|M_{1,6}|=1$. The rank of $M_{1,6}$ is $3$. Hence, $\mathcal{C}^\perp$ has no codeword of weight $3$ whose first nonzero coordinate is at the first $q$ locations and the rest are at the $q+1$-th and $q+2$-th locations.

{Case 1.7:} Let $x$ be an element in $\gf_q$. Consider the submatrix
\begin{eqnarray*}
M_{1,7}=\left[
\begin{array}{lll}
1 & 0 & 1 \\
x & 0 & 0 \\
x^2 & 1 & 1
\end{array}
\right].
\end{eqnarray*}
Then we have $|M_{1,7}|=x$. If $x \neq 0$, the rank of $M_{1,7}$ is $3$. Hence, $\mathcal{C}^\perp$ has no codeword of weight $3$ whose first nonzero coordinate is at the first $q-1$ locations and the other nonzero coordinates are at the $q+1$-th, $q+3$-th locations. If $x=0$, $M_{1,7}$ has the rank $2$. Consequently, the number of codewords of weight $3$ in $\mathcal{C}^\perp$  whose nonzero coordinates are at the $q$-th, $q+1$-th and $q+3$-th locations is equal to $q-1$.

{Case 1.8:} Let $x$ be an element in $\gf_q$. Consider the submatrix
\begin{eqnarray*}
M_{1,8}=\left[
\begin{array}{lll}
1 & 0 & 0 \\
x & 0 & 1 \\
x^2 & 1 & 1
\end{array}
\right].
\end{eqnarray*}
Then we have $|M_{1,8}|=1$. The rank of $M_{1,8}$ is $3$. Hence, $\mathcal{C}^\perp$ has no codeword of weight $3$ whose first nonzero coordinate is at the first $q$ locations and the rest are at the $q+1$-th and $q+4$-th locations.

{Case 1.9:} Let $x$ be an element in $\gf_q$. Consider the submatrix
\begin{eqnarray*}
M_{1,9}=\left[
\begin{array}{lll}
1 & 0 & 1 \\
x & 1 & 0 \\
x^2 & 0 & 1
\end{array}
\right].
\end{eqnarray*}
Then we have $|M_{1,9}|=x^2+1$. $|M_{1,9}|=0$ if and only if $x=1$. If $x=1$, $M_{1,9}$ has the rank $2$. Consequently, the number of codewords of weight $3$ in $\mathcal{C}^\perp$  whose nonzero coordinates are at the first, $q+2$-th and $q+3$-th locations is equal to $q-1$.

{Case 1.10:} Let $x$ be an element in $\gf_q$. Consider the submatrix
\begin{eqnarray*}
M_{1,10}=\left[
\begin{array}{lll}
1 & 0 & 0 \\
x & 1 & 1 \\
x^2 & 0 & 1
\end{array}
\right].
\end{eqnarray*}
Then we have $|M_{1,10}|=1$. The rank of $M_{1,10}$ is $3$. Hence, $\mathcal{C}^\perp$ has no codeword of weight $3$ whose first nonzero coordinate is at the first $q$ locations and the rest are at the $q+2$-th and $q+4$-th locations.

{Case 1.11:} Let $x$ be an element in $\gf_q$. Consider the submatrix
\begin{eqnarray*}
M_{1,11}=\left[
\begin{array}{lll}
1 & 1 & 0 \\
x & 0 & 1 \\
x^2 & 1 & 1
\end{array}
\right].
\end{eqnarray*}
It is easy to derive that $|M_{1,11}|=x^2+x+1$. By Lemma \ref{lem-opoly3}, $|M_{1,11}| \neq 0$. Then the rank of $M_{1,11}$ is $3$. Hence, $\mathcal{C}^\perp$ has no codeword of weight $3$ whose first nonzero coordinate is at the first $q$ locations and the other two nonzero coordinates are at the last two locations.

{Case 1.12:} Consider the submatrix
\begin{eqnarray*}
M_{1,12}=\left[
\begin{array}{lll}
0 & 0 & 1 \\
0 & 1 & 0 \\
1 & 0 & 1
\end{array}
\right].
\end{eqnarray*}
 Clearly, $|M_{1,12}|=1$. Then the rank of $M_{1,12}$ is $3$. Hence, $\mathcal{C}^\perp$ has no codeword of weight $3$ whose nonzero coordinates are at the $q+1$-th, $q+2$-th and $q+3$-th locations.

{Case 1.13:} Consider the submatrix
\begin{eqnarray*}
M_{1,13}=\left[
\begin{array}{lll}
0 & 0 & 0 \\
0 & 1 & 1 \\
1 & 0 & 1
\end{array}
\right].
\end{eqnarray*}
It is easy to find that the rank of $M_{1,13}$ is $2$. Consequently, the number of codewords of weight $3$ in $\mathcal{C}^\perp$ whose nonzero coordinates are at the $q+1$-th, $q+2$-th and $q+4$-th locations is equal to $q-1$.

{Case 1.14:} Consider the submatrix
\begin{eqnarray*}
M_{1,14}=\left[
\begin{array}{lll}
0 & 1 & 0 \\
0 & 0 & 1 \\
1 & 1 & 1
\end{array}
\right].
\end{eqnarray*}
 Since $|M_{1,14}|=1$, the rank of $M_{1,14}$ is $3$. Hence, $\mathcal{C}^\perp$ has no codeword of weight $3$ whose nonzero coordinates are at the $q+1$-th, $q+3$-th and $q+4$-th locations.

{Case 1.15:} Consider the submatrix
\begin{eqnarray*}
M_{1,15}=\left[
\begin{array}{lll}
0 & 1 & 0 \\
1 & 0 & 1 \\
0 & 1 & 1
\end{array}
\right].
\end{eqnarray*}
 Since $|M_{1,15}|=1$, the rank of $M_{1,15}$ is $3$. Hence, $\mathcal{C}^\perp$ has no codeword of weight $3$ whose nonzero coordinates are at the $q+2$-th, $q+3$-th and $q+4$-th locations.

Summarizing the above cases, the total number of codewords of weight $3$ in $\mathcal{C}^\perp$ is $(q-1)(q+2)$.

\subsubsection*{We finally prove that the minimum distance of $\mathcal{C}$ is $q+1$}

Assume that $d(\mathcal{C}) \leq q=q+4-4$ and let $\bc=a \bg_1 + b \bg_2 + c \bg_3$ be a codeword with the minimum weight in $\mathcal{C}$, where $\bg_1$, $\bg_2$ and $\bg_3$ respectively represent the first, second and third rows of $G_1$. We deduce that at least four coordinates  in $\bc$ are zero. Consider the following cases.

{Case 2.1:} Assume that the last four coordinates in $\bc$ are zero.
\begin{eqnarray*}
\left\{
\begin{array}{r}
c=0, \\
b=0, \\
a+c=0, \\
b+c=0.  \\
\end{array}
\right.
\end{eqnarray*}
We deduce that $a=b=c=0$ and $\bc=0$. This is contrary to the fact that $\bc$ is a minimum weight codeword in $\mathcal{C}$.

{Case 2.2:} Assume that three of the last four coordinates in $\bc$ are zero. Then there exists an element $x$ in $\gf_q$ such that
\begin{eqnarray*}
\left\{
\begin{array}{r}
a+bx+cx^2 = 0, \\
c=0, \\
b=0, \\
a+c=0,  \\
\end{array}
\right.
\mbox{ or }
\left\{
\begin{array}{r}
a+bx+cx^2 = 0, \\
c=0, \\
b=0, \\
b+c=0,  \\
\end{array}
\right.
\mbox{ or }
\left\{
\begin{array}{r}
a+bx+cx^2 = 0, \\
c=0, \\
a+c=0, \\
b+c=0,  \\
\end{array}
\right.
\mbox{ or }
\left\{
\begin{array}{r}
a+bx+cx^2 = 0, \\
b=0, \\
a+c=0, \\
b+c=0.  \\
\end{array}
\right.
\end{eqnarray*}
We deduce that $a=b=c=0$ and $\bc=0$. This is contrary to the fact that $\bc$ is a minimum weight codeword in $\mathcal{C}$.

{Case 2.3:} Assume that two of the last four coordinates in $\bc$ are zero. Then there exist two different elements $x$ and $y$ in $\gf_q$ such that
\begin{eqnarray*}
\left\{
\begin{array}{r}
a+bx+cx^2 = 0, \\
a+by+cy^2 = 0, \\
c=0,  \\
b=0, \\
\end{array}
\right.
\mbox{ or }
\left\{
\begin{array}{r}
a+bx+cx^2 = 0, \\
a+by+cy^2 = 0, \\
c=0,  \\
a+c=0, \\
\end{array}
\right.
\mbox{ or }
\left\{
\begin{array}{r}
a+bx+cx^2 = 0, \\
a+by+cy^2 = 0, \\
c=0,  \\
b+c=0, \\
\end{array}
\right.
\end{eqnarray*}
\begin{eqnarray*}
\left\{
\begin{array}{r}
a+bx+cx^2 = 0, \\
a+by+cy^2 = 0, \\
a+c=0,  \\
b+c=0, \\
\end{array}
\right.
\mbox{ or }
\left\{
\begin{array}{r}
a+bx+cx^2 = 0, \\
a+by+cy^2 = 0, \\
b=0,  \\
a+c=0, \\
\end{array}
\right.
\mbox{ or }
\left\{
\begin{array}{r}
a+bx+cx^2 = 0, \\
a+by+cy^2 = 0, \\
b=0,  \\
b+c=0. \\
\end{array}
\right.
\end{eqnarray*}
We deduce that $a=b=c=0$ and $\bc=0$ by Lemma \ref{lem-opoly3}. This is contrary to the fact that $\bc$ is a minimum weight codeword in $\mathcal{C}$.

{Case 2.4:} Assume that at most one of the last four coordinates in $\bc$ is zero. Let $x, y, z$ be three  pairwise different elements in $\gf_q$, we have
\begin{eqnarray*}
\left\{
\begin{array}{r}
a+bx+cx^2 = 0, \\
a+by+cy^2 = 0, \\
a+bz+cz^2 = 0. \\
\end{array}
\right.
\end{eqnarray*}
The coefficient matrix of this system of equations is
\begin{eqnarray*}
M_1 = \left[
\begin{array}{ccc}
1 & x & x^2 \\
1 & y & y^2 \\
1 & z & z^2 \\
\end{array}
\right]
\end{eqnarray*}
Obviously, $|M_1|=(x^2+z^2)(y+z)+(y^2+z^2)(x+z)$. Then by Lemma \ref{lem-opoly2}, $|M_1| \neq 0$. Hence, the rank of the coefficient matrix $M_1$ is $3$, which yields $a=b=c=0$ and $\bc=0$. This is contrary to the fact that $\bc$ is a minimum weight codeword in $\mathcal{C}$.

Summarizing the above discussions, $d(\mathcal{C}) \geq q+1$ has been proved. By the Singleton bound, $d(\mathcal{C}) \leq q+2$. If $d(\mathcal{C}) = q+2$, then $\mathcal{C}$ is an $[q+4, 3, q+2]$ MDS code and $\mathcal{C}^\perp$ is also an MDS code with parameters $[q+4, q+1, 4]$ which is contrary to $d(\mathcal{C}^\perp)=3$. Thus $\mathcal{C}$ is a $[q+4, 3, q+1]$ AMDS code. Besides, $\mathcal{C}$ is an NMDS code as both $\C$ and $\C^\perp$ are AMDS. By Lemma \ref{lem-N2}, the total number $A_{q+1}$ of the minimum weight codewords in $\mathcal{C}$ is equal to the total number of weight $3$ in $\mathcal{C}^\perp$. Hence $A_{q+1}=(q-1)(q+2)$. By Lemma \ref{lem-N1}, the weight enumerator of $\mathcal{C}$ directly follows.
\end{proof}

Below we give another construction of NMDS code with the same parameters $[q+4, 3, q+1]$ but different weight enumerators.
Define
\begin{eqnarray*}
G_{1,1}=\left[
\begin{array}{lllllllll}
1 & 1 & \cdots & 1 & 1 & 0 & 1 & 1 & 0\\
\alpha_1 & \alpha_2 & \cdots & \alpha_{q-1} & 0 & 0 & 0 & 1 & 1\\
\alpha_1^2 & \alpha_2^2 & \cdots & \alpha_{q-1}^2 & 0 & 1 & 1 & 0 & 1
\end{array}
\right].
\end{eqnarray*}
Obviously, $G_{1,1}$ is a $3$ by $q+4$ matrix over $\gf_q$. Let  $\mathcal{C}_1$ be the linear code over $\gf_q$ with generator matrix $G_{1,1}$. We can also derive the parameters and weight enumerators of $\mathcal{C}_1$ with a similar proof to that of Theorem \ref{thm-N4235}. The proof of the following theorem is omitted.

\begin{theorem}\label{thm-N11}
Let $m \geq 3$ be odd. Then the linear code $\mathcal{C}_{1}$ is an NMDS code over $\gf_q$ with parameters $[q+4, 3, q+1]$ and weight enumerator
\begin{eqnarray*}
A(z)=1 + \frac{(q-1)(3q+2)}{2} z^{q+1} + \frac{(q-1)(q^2-2q+6)}{2} z^{q+2} + \\
\frac{5(q-1)(q-2)}{2} z^{q+3} + \frac{(q-1)(q-2)^2}{2} z^{q+4}.
\end{eqnarray*}
\end{theorem}

It is known that any $[n,k,n-k+1]$ MDS code over $\gf_q$ must have a unique weight enumerator. However, this fact is not true for NMDS codes.
The NMDS codes in Theorems  \ref{thm-N4235} and \ref{thm-N11} do have different weight enumerators, though they have the same parameters.
This implies that the NMDS codes in Theorems  \ref{thm-N4235} and \ref{thm-N11} are inequivalent to each other.
\subsection{NMDS code with parameters $[q+3, 3, q]$}
Define
\begin{eqnarray*}
G_2=\left[
\begin{array}{llllllll}
1 & 1 & \cdots & 1 & 1 & 0 & 1 & 0\\
\alpha_1 & \alpha_2 & \cdots & \alpha_{q-1} & 0 & 0 & 0 & 1\\
\alpha_1^2 & \alpha_2^2 & \cdots & \alpha_{q-1}^2 & 0 & 1 & 1 & 1
\end{array}
\right].
\end{eqnarray*}
Obviously, $G_2$ is a $3$ by $q+3$ matrix over $\gf_q$. Let $\mathcal{D}$ be the linear code  over $\gf_q$ with generator matrix $G_2$. Next, we determine the parameters and weight enumerators of $\mathcal{D}$.

\begin{theorem}\label{thm-N4234}
Let $m \geq 3$ be odd. Then the linear code $\mathcal{D}$ is an NMDS code over $\gf_q$ with parameters $[q+3, 3, q]$ and weight enumerator
\begin{eqnarray*}
A(z)=1 + q(q-1) z^{q} + \frac{(q-1)(q^2-q+6)}{2} z^{q+1} + \\
(q-1)(2q-3) z^{q+2} + \frac{(q-1)(q^2-3q+2)}{2} z^{q+3}.
\end{eqnarray*}
\end{theorem}

\begin{proof}
Note that the $q$-th, $q+1$-th and $q+3$-th columns of the generator matrix $G_2$ are linearly independent. Hence
 $\rank(G_2)=3$ implying $\dim(\mathcal{D})=3$.

\subsubsection*{We then prove that $\mathcal{D}^\perp$ has parameters $[q+3, q, 3]$.}

Obviously, $\dim(\mathcal{D}^\perp)=(q+3)-\dim(\mathcal{D})=q$. Since no column of $G_2$ is the zero vector and any two columns of $G_2$ are linearly independent over $\gf_q$, the minimum distance $d(\mathcal{D}^\perp) > 2$. We deduce that $\mathcal{D}^\perp$ has codewords of weight $3$ as the $q$-th, $q+1$-th and $q+3$-th columns of $G_2$ are linearly independent. Then $d(\mathcal{D}^\perp) = 3$. Now we calculate the total number of codewords of weight $3$ in $\mathcal{D}^\perp$. We need to consider the following cases.

{Case 1.1:} Let $x, y, z$ be three pairwise different elements in $\gf_q$. Consider the submatrix
\begin{eqnarray*}
M_{1,1}=\left[
\begin{array}{lll}
1 & 1 & 1 \\
x & y & z \\
x^2 & y^2 & z^2
\end{array}
\right].
\end{eqnarray*}
We have $|M_{1,1}|=(x+y)(x^2+z^2)+(x+z)(x^2+y^2)$. By Lemma \ref{lem-opoly2}, $|M_{1,1}| \neq 0$ and $\rank(M_{1,1})=3$. In conclusion, $\mathcal{D}^\perp$ has no codeword of weight $3$ whose nonzero coordinates are at the first $q$ locations.

{Case 1.2:} Let $x, y$ be two different elements in $\gf_q$. Consider the submatrix
\begin{eqnarray*}
M_{1,2}=\left[
\begin{array}{lll}
1 & 1 & 0 \\
x & y & 0 \\
x^2 & y^2 & 1
\end{array}
\right].
\end{eqnarray*}
Then we have $|M_{1,2}|=y+x$. Since $x \neq y$, $|M_{1,2}| \neq 0$. The rank of $M_{1,2}$ is $3$. Hence, $\mathcal{D}^\perp$ has no codeword of weight $3$ whose the first two nonzero coordinates are at the first $q$ locations and the rest is at the $q+1$-th location.

{Case 1.3:} Let $x, y$ be two different elements in $\gf_q$. Consider the submatrix
\begin{eqnarray*}
M_{1,3}=\left[
\begin{array}{lll}
1 & 1 & 1 \\
x & y & 0 \\
x^2 & y^2 & 1
\end{array}
\right].
\end{eqnarray*}
It is easy to deduce that $|M_{1,3}|=(y^2+1)x+(x^2+1)y$. If $x=0, y=1$ or $x=1, y=0$, then $|M_{1,3}| \neq 0$. The rank of $M_{1,3}$ is $3$. Hence, $\mathcal{D}$ does not have a codeword of weight $3$ whose coordinates are at the first, $q$-th and $q+2$-th locations. Next we count the number of different pair $(x,y)$ such that $|M_{1,3}|=0$, where $x,y \in \gf_q \setminus \{0,1\}$ . For any $x,y \in \gf_q \setminus \{0,1\}$, let $a=\frac{x^2+1}{x}$. Then $a \neq 0$. By Lemma \ref{lem-opoly3}, $a \neq 1$. Note that $z^2+az$ is $2$-to-$1$ by Lemma \ref{lem-opoly1}. Therefore, there exists another element $y \in \gf_q \setminus \{0,1\}$ such that $x^2+ax=1=y^2+ay$. For this pair $(x, y)$ we have $|M_{1,3}|=0$ and  vice versa. It follows that the number of distinct $(x,y) \in \gf_q \setminus \{0,1\}$ such that $|M_{1,3}|=0$ is equal to $(q-2)/2$. Consequently, the number of codewords of weight $3$ in $\mathcal{D}^\perp$  whose first two nonzero coordinates are at the first $q$ locations (expect the first and $q$-th locations) and the rest is at the $q+2$-th location is equal to $(q-2)(q-1)/2$.

{Case 1.4:} Let $x, y$ be two distinct elements in $\gf_q$. Consider the submatrix
\begin{eqnarray*}
M_{1,4}=\left[
\begin{array}{lll}
1 & 1 & 0 \\
x & y & 1 \\
x^2 & y^2 & 1
\end{array}
\right].
\end{eqnarray*}
It is easy to deduce that $|M_{1,4}|=x^2+y^2+x+y$. Now we calculate the number of $(x,y)$ satisfying $|M_{1,4}|=0$. Let $|M_{1,4}|=x^2+y^2+x+y=0$ which is equivalent to $(x+y)^2=x+y$. Since $x \neq y$, this equation can be simplified to $x+y=1$. We deduce that the number of different $(x,y)$ such that $|M_{1,4}|=0$ is equal to $q/2$. In conclusion, the number of codewords of weight $3$ in $\mathcal{D}^\perp$ whose first two nonzero coordinates are at the first $q$ locations and the rest is at the $q+3$-th location is equal to $q(q-1)/2$.

{Case 1.5:} Let $x$ be an element in $\gf_q$. Consider the submatrix
\begin{eqnarray*}
M_{1,5}=\left[
\begin{array}{lll}
1 & 0 & 1 \\
x & 0 & 0 \\
x^2 & 1 & 1
\end{array}
\right].
\end{eqnarray*}
It is easy to derive $|M_{1,5}|=x$. If $x \neq 0$, then $\rank(M_{1,5})=3$. Hence $\mathcal{D}^\perp$ does not have a codeword of weight $3$ whose first nonzero coordinate is at the first $q-1$ locations and the other nonzero coordinates are at the last two locations. If $x=0$, then $\rank(M_{1,5})=2$. Consequently, the number of codewords of weight $3$ in $\mathcal{D}^\perp$  whose nonzero coordinates are at the $q$-th, $q+1$-th, and $q+2$-th locations is equal to $q-1$.

{Case 1.6:} Let $x$ be an element in $\gf_q$. Consider the submatrix
\begin{eqnarray*}
M_{1,6}=\left[
\begin{array}{lll}
1 & 0 & 0 \\
x & 0 & 1 \\
x^2 & 1 & 1
\end{array}
\right].
\end{eqnarray*}
It is easy to prove that $|M_{1,6}|=1$ and $\rank(M_{1,6})=3$. Hence, $\mathcal{D}^\perp$ has no codeword of weight $3$ whose first nonzero coordinate is at the first $q$ locations and the other nonzero coordinates are at the $q+1$-th and $q+3$-th locations.

{Case 1.7:} Let $x$ be an element in $\gf_q$. Consider the submatrix
\begin{eqnarray*}
M_{1,7}=\left[
\begin{array}{lll}
1 & 1 & 0 \\
x & 0 & 1 \\
x^2 & 1 & 1
\end{array}
\right].
\end{eqnarray*}
 We have $|M_{1,7}|=x^2+x+1$. By Lemma \ref{lem-opoly3}, $|M_{1,7}| \neq 0$ and  $\rank(M_{1,7})=3$. Hence, $\mathcal{D}^\perp$ has no codeword of weight $3$ whose first nonzero coordinate is at the first $q$ locations and the other two nonzero coordinates are at the last two locations.

{Case 1.8:} Let $x$ be an element in $\gf_q$. Consider the submatrix
\begin{eqnarray*}
M_{1,8}=\left[
\begin{array}{lll}
0 & 1 & 0 \\
0 & 0 & 1 \\
1 & 1 & 1
\end{array}
\right].
\end{eqnarray*}
We have $|M_{1,8}|=1$. Then the rank of $M_{1,8}$ is $3$. Hence, $\mathcal{D}^\perp$ has no codeword of weight $3$ whose nonzero coordinates are at the last three locations.

Summarizing the above eight cases, we deduce that the total number of codewords of weight $3$ in $\mathcal{D}^\perp$ is $q(q-1)$.

\subsubsection*{We finally prove that the minimum distance of $\mathcal{D}$ is $q$}

Assume that $d(\mathcal{D}) \leq q-1=q+3-4$. Let $\bc=a \bg_1 + b \bg_2 + c \bg_3$ be a codeword with the minimum weight in $\mathcal{D}$, where $\bg_1$, $\bg_2$ and $\bg_3$ respectively represent the first, second and third rows of $G_2$. Then at least four coordinates are zero in $\bc$. Consider the following three cases.

{Case 2.1:} Assume that the last three coordinates in $\bc$ are zero. Then there exists an element $x$ in $\gf_q$ such that
\begin{eqnarray*}
\left\{
\begin{array}{r}
a+bx+cx^2 = 0, \\
a+c=0, \\
b+c=0,  \\
c=0. \\
\end{array}
\right.
\end{eqnarray*}
We deduce that $a=b=c=0$ and $\bc=0$. This is contrary to the fact that $\bc$ is a minimum weight codeword in $\mathcal{D}$.

{Case 2.2:} Assume that two of the last three coordinates in $\bc$ are zero. Then there exist two different elements $x$ and $y$ in $\gf_q$ such that
\begin{eqnarray*}
\left\{
\begin{array}{r}
a+bx+cx^2 = 0, \\
a+by+cy^2 = 0, \\
c=0, \\
a+c=0,  \\
\end{array}
\right.
\mbox{ or }
\left\{
\begin{array}{r}
a+bx+cx^2 = 0, \\
a+by+cy^2 = 0, \\
c=0, \\
b+c=0,  \\
\end{array}
\right.
\mbox{ or }
\left\{
\begin{array}{r}
a+bx+cx^2 = 0, \\
a+by+cy^2 = 0, \\
a+c=0, \\
b+c=0.  \\
\end{array}
\right.
\end{eqnarray*}
We deduce that $a=b=c=0$ and $\bc=0$ by Lemma \ref{lem-opoly3}. This is contrary to the fact that $\bc$ is a minimum weight codeword in $\mathcal{D}$.

{Case 2.3:} Assume that at most one of the last three coordinates in $\bc$ is zero. Then there exist three pairwise different elements $x, y, z$ in $\gf_q$ such that
\begin{eqnarray*}
\left\{
\begin{array}{r}
a+bx+cx^2 = 0, \\
a+by+cy^2 = 0, \\
a+bz+cz^2 = 0. \\
\end{array}
\right.
\end{eqnarray*}
The coefficient matrix for this system of equations is
\begin{eqnarray*}
M_1 = \left[
\begin{array}{ccc}
1 & x & x^2 \\
1 & y & y^2 \\
1 & z & z^2 \\
\end{array}
\right]
\end{eqnarray*}
Obviously, $|M_1|=(x^2+z^2)(y+z)+(y^2+z^2)(x+z)$. By Lemma \ref{lem-opoly2}, $|M_1| \neq 0$. Hence, the rank of the coefficient matrix $M_1$ is $3$ implying $a=b=c=0$ and $\bc=0$. This is contrary to the fact that $\bc$ is a minimum weight codeword in $\mathcal{D}$.

Summarizing the above discussions, $d(\mathcal{D}) \geq q$ has been proved. By the Singleton bound, $d(\mathcal{D}) \leq q+1$. If $d(\mathcal{D}) = q+1$, then $\mathcal{D}$ is a $[q+3, 3, q+1]$ MDS code and $\mathcal{D}^\perp$ is also an MDS code with parameters $[q+3, q, 4]$, which is contrary to $d(\mathcal{D}^\perp)=3$. Finally we have $d(\mathcal{D})=q$. Thus, $\mathcal{D}$ is a $[q+3, 3, q]$ NMDS code. By Lemma \ref{lem-N2}, the total number $A_{q}$ of the minimum weight codewords in $\mathcal{D}$ is equal to the total number of weight $3$ in $\mathcal{D}^\perp$. Hence $A_{q}=q(q-1)$. The weight enumerator of $\mathcal{D}$ follows from Lemma \ref{lem-N1}.
\end{proof}

Below we give two  other  constructions of NMDS code with the same parameters $[q+3, 3, q]$ but different weight enumerators.
Define
\begin{eqnarray*}
G_{2,1}=\left[
\begin{array}{llllllll}
1 & 1 & \cdots & 1 & 1 & 0 & 0 & 1\\
\alpha_1 & \alpha_2 & \cdots & \alpha_{q-1} & 0 & 1 & 0 & 1\\
\alpha_1^2 & \alpha_2^2 & \cdots & \alpha_{q-1}^2 & 0 & 0 & 1 & 0
\end{array}
\right].
\end{eqnarray*}
Obviously, $G_{2,1}$ is a $3$ by $q+3$ matrix over $\gf_q$. Let $\mathcal{D}_{1}$ be the linear code  over $\gf_q$ with  generator matrix $G_{2,1}$.

 With a similar proof to that of Theorem \ref{thm-N4234}, we can derive the parameters and weight enumerator of $\mathcal{D}_{1}$ in the following theorem.

\begin{theorem}\label{thm-N12}
Let $m \geq 2$ be an integer. Then the linear code $\mathcal{D}_{1}$ is an NMDS code over $\gf_q$ with parameters $[q+3, 3, q]$ and weight enumerator
\begin{eqnarray*}
A(z)=1 + \frac{(q-1)(q+2)}{2} z^{q} + \frac{q(q-1)(q+2)}{2} z^{q+1} + \\
\frac{q(q-1)}{2} z^{q+2} + \frac{q(q-1)(q-2)}{2} z^{q+3}.
\end{eqnarray*}
\end{theorem}

We remark that the NMDS code in Theorem \ref{thm-N12} has the same weight enumerator as that of the NMDS code in \cite[Theorem 8]{Wqiu} for $f(x)=x^2$.
These two NMDS codes have different generator matrixes. It is open whether they are equivalent to each other.  Besides, Theorem \ref{thm-N12} holds for any integer $m\geq 2$ as its proof dose not rely on Lemma \ref{lem-opoly3}.

Define
\begin{eqnarray*}
G_{2,2}=\left[
\begin{array}{llllllll}
1 & 1 & \cdots & 1 & 1 & 1 & 0 & 1\\
\alpha_1 & \alpha_2 & \cdots & \alpha_{q-1} & 0 & 0 & 1 & 1\\
\alpha_1^2 & \alpha_2^2 & \cdots & \alpha_{q-1}^2 & 0 & 1 & 1 & 0
\end{array}
\right].
\end{eqnarray*}
Obviously, $G_{2,2}$ is a $3$ by $q+3$ matrix over $\gf_q$. Let $\mathcal{D}_{2}$ be the linear code over $\gf_q$ with generator matrix $G_{2,2}$.

 With a similar proof to that of Theorem \ref{thm-N4234}, we can derive the parameters and weight enumerator of $\mathcal{D}_{2}$ in the following theorem.

\begin{theorem}\label{thm-N13}
Let $m \geq 3$ be odd. Then the linear code $\mathcal{D}_{2}$ is an NMDS code over $\gf_q$ with parameters $[q+3, 3, q]$ and weight enumerator
\begin{eqnarray*}
A(z)=1 + \frac{(q-1)(3q-2)}{2} z^{q} + \frac{(q-1)(q^2-4q+12)}{2} z^{q+1} + \\
\frac{(q-1)(7q-12)}{2} z^{q+2} + \frac{(q-1)(q-2)^2}{2} z^{q+3}.
\end{eqnarray*}
\end{theorem}

Note that the NMDS codes in Theorems \ref{thm-N4234}, \ref{thm-N12} and \ref{thm-N13} have different weight enumerators for odd $m \geq 3$, though they have the same parameters. Hence they are pairwise inequivalent to each other. Besides, the NMDS codes in Theorems \ref{thm-N4234} and \ref{thm-N13} are inequivalent to the NMDS code in \cite[Theorem 8]{Wqiu} due to different weight enumerators.
\subsection{NMDS codes with parameters $[q+1, 3, q-2]$}
Define
\begin{eqnarray*}
G_3=\left[
\begin{array}{llllll}
1 & 1 & \cdots & 1 & 1 & 0 \\
\alpha_1 & \alpha_2 & \cdots & \alpha_{q-1} & 0 & 1 \\
\alpha_1^2 & \alpha_2^2 & \cdots & \alpha_{q-1}^2 & 0 & 1
\end{array}
\right].
\end{eqnarray*}
Obviously, $G_3$ is a $3$ by $q+1$ matrix over $\gf_q$. Let $\mathcal{E}$  be  the linear code over $\gf_q$ with generator matrix $G_3$. Next, we calculate the parameters and weight enumerator of $\mathcal{E}$.

\begin{theorem}\label{thm-N1235}
Let $m \geq 2$ be an integer. Then the linear code $\mathcal{E}$ is an NMDS code over $\gf_q$ with parameters $[q+1, 3, q-2]$ and weight enumerator
\begin{eqnarray*}
A(z)=1 + \frac{q(q-1)}{2}z^{q-2} + \frac{q(q-1)(q-2)}{2} z^{q-1} + \\
\frac{(q-1)(5q+2)}{2} z^{q} + \frac{q(q-1)(q-2)}{2} z^{q+1}.
\end{eqnarray*}
\end{theorem}

\begin{proof}
We first prove $\dim(\mathcal{E})=3$. Let $\bg_1$, $\bg_2$ and $\bg_3$ represent the first, second and third rows of $G_3$. Assume that there exist three elements $a$, $b$ and $c$ in $\gf_q$ such that $a \bg_1 + b \bg_2 + c \bg_3 =0$, where at least one of the elements in $\{a,b,c\}$ is nonzero. Then we have
\begin{eqnarray*}
\left\{
\begin{array}{l}
a=0, \\
b+c=0, \\
a+bx+cx^2 = 0 \mbox{ for all } x \in \gf_q^*. \\
\end{array}
\right.
\end{eqnarray*}
It is easy to derive that $a=b=c=0$ and $\dim(\mathcal{E})=3$.

\subsubsection*{We then prove that $\mathcal{E}^\perp$ has parameters $[q+1, q-2, 3]$.}

Obviously, $\dim(\mathcal{E}^\perp)=(q+1)-\dim(\mathcal{E})=q-2$. It is easy to find that  no column of $G_3$ is the zero vector and any two columns of $G_3$ are $\gf_q$-linearly independent. Then the minimum distance $d(\mathcal{E}^\perp) > 2$. We also find the first, $q$-th, $q+1$-th columns of $G_3$ are linearly dependent, which means that $\mathcal{E}^\perp$ has codewords of weight $3$. Then we have the minimum distance $d(\mathcal{E}^\perp)=3$. Now we calculate the total number of codewords of weight $3$ in $\mathcal{E}^\perp$. We need to consider the following two cases.

{Case 1.1:} Let $x, y, z$ be three pairwise different elements in $\gf_q$. Consider the submatrix
\begin{eqnarray*}
M_{1,1}=\left[
\begin{array}{lll}
1 & 1 & 1 \\
x & y & z \\
x^2 & y^2 & z^2
\end{array}
\right].
\end{eqnarray*}
We have $|M_{1,1}|=(x+y)(x^2+z^2)+(x+z)(x^2+y^2)$. By Lemma \ref{lem-opoly2}, $|M_{1,1}| \neq 0$. Hence, the rank of $M_{1,1}$ is three. In conclusion, $\mathcal{E}^\perp$ has no codeword of weight $3$ whose nonzero coordinates are at the first $q$ locations.

{Case 1.2:} Let $x, y$ be two distinct elements in $\gf_q$. Consider the submatrix
\begin{eqnarray*}
M_{1,2}=\left[
\begin{array}{lll}
1 & 1 & 0 \\
x & y & 1 \\
x^2 & y^2 & 1
\end{array}
\right].
\end{eqnarray*}
It is easy to deduce that $|M_{1,2}|=x^2+y^2+x+y$. Now we calculate the number of $(x,y)$ such that $|M_{1,2}|=0$. Let $|M_{1,2}|=x^2+y^2+x+y=0$ which is equivalent to $(x+y)^2=x+y$. Since $x \neq y$, then this equation  can be simplified to $x+y=1$. We then deduce that the number of different $(x,y)$ such that $|M_{1,2}|=0$ is equal to $q/2$. In conclusion, the number of codewords of weight $3$ in $\mathcal{E}^\perp$ whose first two nonzero coordinates are at the first $q$ locations and the rest is at the $q+1$-th location is equal to $q(q-1)/2$.

Summarizing the above cases, the total number of codewords of weight $3$ in $\mathcal{E}^\perp$ is $q(q-1)/2$.

\subsubsection*{We finally prove that the minimum distance of $\mathcal{E}$ is $q-2$}

Assume that $d(\mathcal{E}) \leq q-3=q+1-4$ and let $\bc=a \bg_1 + b \bg_2 + c \bg_3$ be a codeword with the minimum weight in $\mathcal{E}$, where $\bg_1$, $\bg_2$ and $\bg_3$ respectively represent the first, second and third rows of $G_3$. It can be deduced that there exist at least four coordinates  in $\bc$ are zero. Then we can find that there exist at least three coordinates at the first $q$ locations in $\bc$. Let $x, y, z$ be three  pairwise different elements in $\gf_q$, we have
\begin{eqnarray*}
\left\{
\begin{array}{r}
a+bx+cx^2 = 0, \\
a+by+cy^2 = 0, \\
a+bz+cz^2 = 0. \\
\end{array}
\right.
\end{eqnarray*}
The coefficient matrix of this system of equations is
\begin{eqnarray*}
M_1 = \left[
\begin{array}{ccc}
1 & x & x^2 \\
1 & y & y^2 \\
1 & z & z^2 \\
\end{array}
\right]
\end{eqnarray*}
Obviously, $|M_1|=(x^2+z^2)(y+z)+(y^2+z^2)(x+z)$. Then by Lemma \ref{lem-opoly2}, $|M_1| \neq 0$. Hence, the rank of the coefficient matrix $M_1$ is $3$, which yields $a=b=c=0$ and $\bc=0$. This is contrary to the fact that $\bc$ is a minimum weight codeword in $\mathcal{E}$.

Then $d(\mathcal{E}) \geq q-2$ has been proved. By the Singleton bound, $d(\mathcal{C}) \leq q-1$. If $d(\mathcal{E}) = q-1$, then $\mathcal{E}$ is an $[q+1, 3, q-1]$ MDS code and $\mathcal{E}^\perp$ is also an MDS code with parameters $[q+1, q-2, 4]$ which is contrary to $d(\mathcal{E}^\perp)=3$. Thus $\mathcal{E}$ is a $[q+1, 3, q-2]$ AMDS code. Besides, $\mathcal{E}$ is an NMDS code as both $\mathcal{E}$ and $\mathcal{E}^\perp$ are AMDS. By Lemma \ref{lem-N2}, the total number $A_{q-2}$ of the minimum weight codewords in $\mathcal{E}$ is equal to the total number of weight $3$ in $\mathcal{E}^\perp$. Hence $A_{q-2}=\frac{q(q-1)}{2}$. By Lemma \ref{lem-N1}, the weight enumerator of $\mathcal{E}$ directly follows.
\end{proof}

Below we give two  other   constructions of NMDS codes with the same parameters $[q+1, 3, q-2]$ but different weight enumerators.
Define
\begin{eqnarray*}
G_{3,1}=\left[
\begin{array}{llllll}
1 & 1 & \cdots & 1 & 1 & 0 \\
\alpha_1 & \alpha_2 & \cdots & \alpha_{q-1} & 1 & 1 \\
\alpha_1^2 & \alpha_2^2 & \cdots & \alpha_{q-1}^2 & 0 & 1
\end{array}
\right].
\end{eqnarray*}
Obviously, $G_{3,1}$ is a $3$ by $q+1$ matrix over $\gf_q$. Let $\mathcal{E}_1$  be  the linear code over $\gf_q$ with generator matrix $G_{3,1}$.

With a similar proof to that of Theorem \ref{thm-N1235}, we can derive the parameters and weight enumerator of $\mathcal{E}_1$  in the following.

\begin{theorem}\label{thm-N1234}
Let $m \geq 3$ be odd. Then the linear code $\mathcal{E}_1$ is an NMDS code over $\gf_q$ with parameters $[q+1, 3, q-2]$ and weight enumerator
\begin{eqnarray*}
A(z)=1 + (q-1)(q-2)z^{q-2} + \frac{(q-1)(q^2-5q+12)}{2} z^{q-1} + \\
(q-1)(4q-5) z^{q} + \frac{(q-1)(q^2-3q+4)}{2} z^{q+1}.
\end{eqnarray*}
\end{theorem}

We remark that the NMDS code in Theorem \ref{thm-N1234} has the same weight enumerator as that of the NMDS code in \cite[Theorem 10]{Wqiu} for $f(x)=x^2$.
These two NMDS codes have different generator matrixes. It is open whether they are equivalent to each other. Besides, Theorem \ref{thm-N1234} holds for any odd integer $m \geq 3$ as its proof relies on Lemma \ref{lem-opoly3}.

Define
\begin{eqnarray*}
G_{3,2}=\left[
\begin{array}{llllll}
1 & 1 & \cdots & 1 & 1 & 1 \\
\alpha_1 & \alpha_2 & \cdots & \alpha_{q-1} & 0 & 0 \\
\alpha_1^2 & \alpha_2^2 & \cdots & \alpha_{q-1}^2 & 0 & 1
\end{array}
\right].
\end{eqnarray*}
Obviously, $G_{3,2}$ is a $3$ by $q+1$ matrix over $\gf_q$. Let $\mathcal{E}_2$  be  the linear code over $\gf_q$ with generator matrix $G_{3,2}$.

With a similar proof to that of Theorem \ref{thm-N1235}, we can determine the parameters and weight enumerator of $\mathcal{E}_2$  in the following.

\begin{theorem}\label{thm-N17}
Let $m \geq 2$ be an integer. Then the linear code $\mathcal{E}_2$ is an NMDS code over $\gf_q$ with parameters $[q+1, 3, q-2]$ and weight enumerator
\begin{eqnarray*}
A(z)=1 + \frac{(q-1)(q-2)}{2}z^{q-2} + \frac{(q-1)(q^2-2q+6)}{2} z^{q-1} + \\
\frac{(q-1)(5q-4)}{2} z^{q} + \frac{(q-1)(q^2-2q+2)}{2} z^{q+1}.
\end{eqnarray*}
\end{theorem}

The NMDS codes in Theorems \ref{thm-N1235}, \ref{thm-N1234} and \ref{thm-N17} have different weight enumerators, though they have the same parameters. Hence they are pairwise inequivalent to each other. Besides, the NMDS codes in Theorems \ref{thm-N1235} and \ref{thm-N17} are inequivalent to the NMDS code in \cite[Theorem 10]{Wqiu} because of different weight enumerators.

\subsection{NMDS code with parameters $[q+2, 3, q-1]$}

In this subsection, we first introduce the definition of the extended code of a linear code. Let $G$ be the generator matrix of a linear code $\mathcal{C}$. Define a matrix $\bar{G}$  by adding a column  to $G$ such that the sum of the elements of each row of $\bar{G}$ is $0$. The linear code  with generator matrix $\bar{G}$ is called the extended code of $\mathcal{C}$. The extended code of $\mathcal{C}$ is denoted as $\bar{\mathcal{C}}$.

 In the following, we study the extended code of $\mathcal{E}_1$ in Theorem \ref{thm-N1234}.

\begin{theorem}\label{thm-N2234}
Let $m \geq 3$ be odd. Then the extended code $\bar{\mathcal{E}_1}$ is an NMDS code over $\gf_q$ with parameters $[q+2, 3, q-1]$ and weight enumerator
\begin{eqnarray*}
A(z)=1 + (q-1)^{2} z^{q-1} + \frac{(q-1)(q^2-3q+8)}{2} z^{q} + \\
3(q-1)^{2} z^{q+1} + \frac{(q-1)(q^2-3q+2)}{2} z^{q+2}.
\end{eqnarray*}
\end{theorem}

\begin{proof}
It is well known that $\mathop{\sum}_{x \in \gf_q} x =0$. Since $f(x)=x^2$ is a permutation polynomial over $\gf_q$, we have $\mathop{\sum}_{x \in \gf_q} x^2=0$. Then the extended code $\bar{\mathcal{E}_1}$ has generator matrix
\begin{eqnarray*}
\bar{G_3}=\left[
\begin{array}{lllllll}
1 & 1 & \cdots & 1 & 1 & 0 & 0 \\
\alpha_1 & \alpha_2 & \cdots & \alpha_{q-1} & 1 & 1 & 0 \\
\alpha_1^2 & \alpha_2^2 & \cdots & \alpha_{q-1}^2 & 0 & 1 & 1
\end{array}
\right].
\end{eqnarray*}

By definition, $\bar{\mathcal{E}_1}$ has length $q+2$ and dimension $3$. Then the dual code $\bar{\mathcal{E}_1}^\perp$ of $\bar{\mathcal{E}_1}$ has length $q+2$ and dimension $q-1$. We now prove the minimum distance of $\bar{\mathcal{E}_1}^\perp$ is $3$. We can find the first, $q$-th, $q+2$-th columns of $\bar{G_3}$ are linearly dependent, which means that $\bar{\mathcal{E}_1}^\perp$ has codewords of weight $3$. Then we have the minimum distance $d(\bar{\mathcal{E}_1}^\perp) \leq 3$. It is easy to deduce that any two columns of $\bar{G_3}$ are linearly independent, then we have $d(\bar{\mathcal{E}_1}^\perp) > 2$. Thus, $d(\bar{\mathcal{E}_1}^\perp)=3$. We then calculate the total number of codewords with weight $3$ in $\bar{\mathcal{E}_1}^\perp$ in the following cases.

{Case 1.1:} Let $x, y, z$ be three pairwise different elements in $\gf_q^*$. Consider the submatrix
\begin{eqnarray*}
M_{1,1}=\left[
\begin{array}{lll}
1 & 1 & 1 \\
x & y & z \\
x^2 & y^2 & z^2
\end{array}
\right].
\end{eqnarray*}
Then we have $|M_{1,1}|=(x+y)(x^2+z^2)+(x+z)(x^2+y^2)$. By Lemma \ref{lem-opoly2}, $|M_{1,1}| \neq 0$. Hence, the rank of $M_{1,1}$ is three. In conclusion, $\bar{\mathcal{E}_1}^\perp$ has no codeword of weight $3$ whose nonzero coordinates are at the first $q-1$ locations.

{Case 1.2:} Let $x, y$ be two different elements in $\gf_q^*$. Consider the submatrix
\begin{eqnarray*}
M_{1,2}=\left[
\begin{array}{lll}
1 & 1 & 1 \\
x & y & 1 \\
x^2 & y^2 & 0
\end{array}
\right].
\end{eqnarray*}
Then we have $|M_{1,2}|=(x+1)y^2+(y+1)x^2$. Choose any $y \in \gf_q \setminus \{0,1\}$. Define $a=y^2/(y+1)$. Then $a \neq 0$. By Lemma \ref{lem-opoly3}, $a \neq 1$. Note that $y^2+ay=a$. By Lemma \ref{lem-opoly1}, $z^2+az$ is $2$-to-$1$. Therefore, there exists another element $x \in \gf_q^*$ such that $x^2+ax=a$. For this pair $(x, y)$ we have $|M_{1,2}|= 0$ and vice versa. Hence, the number of different $(x,y)$ in $\gf_q^*$ such that $|M_{1,2}|=0$ is equal to $(q-2)/2$. Consequently, the number of codewords of weight $3$ in $\bar{\mathcal{E}_1}^\perp$  whose two nonzero coordinates are at the first $q-1$ locations (expect the first location) and the other nonzero coordinate is at the $q$-th location is equal to $(q-2)(q-1)/2$.

{Case 1.3:} Let $x, y$ be two different elements in $\gf_q^*$. Consider the submatrix
\begin{eqnarray*}
M_{1,3}=\left[
\begin{array}{lll}
1 & 1 & 0 \\
x & y & 1 \\
x^2 & y^2 & 1
\end{array}
\right].
\end{eqnarray*}
It is easy to deduce that $|M_{1,3}|=x^2+y^2+x+y$. Now we calculate the number of $(x,y)$ satisfying $|M_{1,3}|=0$. Let $|M_{1,3}|=x^2+y^2+x+y=0$ which is equivalent to $(x+y)^2=x+y$. Since $x \neq y$, this equation can be simplified to $x+y=1$. Then we deduce that the total number of different $x$ and $y$ in $\gf_q^*$ such that $|M_{1,3}|=0$ is equal to $(q-2)/2$ as $x\neq 1$. In conclusion, the number of codewords of weight $3$ in $\bar{\mathcal{E}_1}^\perp$  whose first two nonzero coordinates are at the first $q-1$ locations (expect the first location) and the rest is at the $q+1$-th location is equal to $(q-2)(q-1)/2$.

{Case 1.4:} Let $x, y$ be two different elements in $\gf_q^*$. Consider the submatrix
\begin{eqnarray*}
M_{1,4}=\left[
\begin{array}{lll}
1 & 1 & 0 \\
x & y & 0 \\
x^2 & y^2 & 1
\end{array}
\right].
\end{eqnarray*}
It is easy to deduce that $|M_{1,4}|=y+x$. Since $x \neq y$, $|M_{1,4}| \neq 0$. Hence, $\bar{\mathcal{E}_1}^\perp$ has no codeword of weight $3$ whose first two nonzero coordinates are at the first $q-1$ locations and the rest is at the $q+2$-th location.

{Case 1.5:} Let $x$ be an element in $\gf_q^*$. Consider the submatrix
\begin{eqnarray*}
M_{1,5}=\left[
\begin{array}{lll}
1 & 1 & 0 \\
x & 1 & 1 \\
x^2 & 0 & 1
\end{array}
\right].
\end{eqnarray*}
It is clear that $|M_{1,5}|=x^2+x+1$. By Lemma \ref{lem-opoly3}, $|M_{1,5}| \neq 0$.  Hence, $\bar{\mathcal{E}_1}^\perp$ does not have a codeword of weight $3$ whose first nonzero coordinate is at the first $q-1$ locations and the other two nonzero coordinates are at the $q$-th and $q+1$-th locations.

{Case 1.6:} Let $x$ be an element in $\gf_q^*$. Consider the submatrix
\begin{eqnarray*}
M_{1,6}=\left[
\begin{array}{lll}
1 & 1 & 0 \\
x & 1 & 0 \\
x^2 & 0 & 1
\end{array}
\right].
\end{eqnarray*}
 Clearly, $|M_{1,6}|=1+x$. $|M_{1,6}| = 0$ if and only if $x=1$. If $x=1$, the rank of $M_{1,6}$ is $2$. Consequently, the number of codewords of weight $3$ in $\bar{\mathcal{E}_1}^\perp$ whose nonzero coordinates are at the first, $q$-th and $q+2$-th locations is equal to $q-1$.

{Case 1.7:} Let $x$ be an element in $\gf_q^*$. Consider the submatrix
\begin{eqnarray*}
M_{1,7}=\left[
\begin{array}{lll}
1 & 0 & 0 \\
x & 1 & 0 \\
x^2 & 1 & 1
\end{array}
\right].
\end{eqnarray*}
 Note that $|M_{1,7}|=1$. Hence, $\bar{\mathcal{E}_1}^\perp$ has no codeword of weight $3$ whose first nonzero coordinate is at the first $q-1$ locations and the other nonzero coordinates are at the last two locations.

{Case 1.8:} Consider the submatrix
\begin{eqnarray*}
M_{1,8}=\left[
\begin{array}{lll}
1 & 0 & 0 \\
1 & 1 & 0 \\
0 & 1 & 1
\end{array}
\right].
\end{eqnarray*}
Clearly, $|M_{1,8}|=1$.  Hence, $\bar{\mathcal{E}_1}^\perp$ has no codeword of weight $3$ whose nonzero coordinates are at the last three locations.

Summarizing the above eight cases, the total number of codewords of weight $3$ in $\bar{\mathcal{E}_1}^\perp$ is $(q-1)^2$.

Finally, we prove that the minimum distance of $\bar{\mathcal{E}_1}$ is $q-1$. By definition, $d(\bar{\mathcal{E}_1})=d(\mathcal{E}_1)=q-2$ or $d(\bar{\mathcal{E}_1})=d(\mathcal{E}_1)+1=q-1$. Assume that $d(\bar{\mathcal{E}_1})=q-2$. Let $\bc=a \bg_1 + b \bg_2 + c \bg_3$ be a codeword with weight $q-2$ in $\bar{\mathcal{E}_1}$, where $\bg_1$, $\bg_2$ and $\bg_3$ respectively represent the first, second and third rows of $\bar{G_3}$. Then $\bc$ has four zero coordinates. Consider the following cases.

{Case 2.1:}Assume that the last three coordinates in $\bc$ are zero. Then there exists an element $x$ in $\gf_q^*$ such that
\begin{eqnarray*}
\left\{
\begin{array}{r}
a+bx+cx^2 = 0, \\
a+b=0, \\
b+c=0,  \\
c=0. \\
\end{array}
\right.
\end{eqnarray*}
We deduce that $a=b=c=0$ and $\bc=0$. This is contrary to the fact that $\bc$ is a minimum weight codeword in $\bar{\mathcal{E}_1}$.

{Case 2.2:} Assume that two of the last three coordinates in $\bc$ are zero. Then there exist two different elements $x$ and $y$ in $\gf_q^*$ such that
\begin{eqnarray*}
\left\{
\begin{array}{r}
a+bx+cx^2 = 0, \\
a+by+cy^2 = 0, \\
a+b=0, \\
b+c=0,  \\
\end{array}
\right.
\mbox{ or }
\left\{
\begin{array}{r}
a+bx+cx^2 = 0, \\
a+by+cy^2 = 0, \\
a+b=0, \\
c=0,  \\
\end{array}
\right.
\mbox{ or }
\left\{
\begin{array}{r}
a+bx+cx^2 = 0, \\
a+by+cy^2 = 0, \\
b+c=0,  \\
c=0.
\end{array}
\right.
\end{eqnarray*}
We deduce that $a=b=c=0$ and $\bc=0$ by Lemma \ref{lem-opoly3}. This is contrary to the fact that $\bc$ is a minimum weight codeword in $\bar{\mathcal{E}_1}$.

{Case 2.3:} Assume that at most one of the last three coordinates in $\bc$ is zero. Then there exist three pairwise different elements $x, y, z$  in $\gf_q^*$ such that
\begin{eqnarray*}
\left\{
\begin{array}{r}
a+bx+cx^2 = 0, \\
a+by+cy^2 = 0, \\
a+bz+cz^2 = 0. \\
\end{array}
\right.
\end{eqnarray*}
The coefficient matrix for this system of equations is
\begin{eqnarray*}
M_1 = \left[
\begin{array}{ccc}
1 & x & x^2 \\
1 & y & y^2 \\
1 & z & z^2 \\
\end{array}
\right]
\end{eqnarray*}
Obviously, $|M_1|=(x^2+z^2)(y+z)+(y^2+z^2)(x+z)$. Then by Lemma \ref{lem-opoly2}, $|M_1| \neq 0$. Hence, $a=b=c=0$ and $\bc=0$. This is contrary to the fact that $\bc$ is a minimum weight codeword in $\bar{\mathcal{E}_1}$.

Summarizing the above discussions, we have $d(\bar{\mathcal{E}_1})=d(\mathcal{E}_1)+1=q-1$. Hence $\bar{\mathcal{E}_1}$ is an NMDS code with parameters $[q+2, 3, q-1]$. By Lemma \ref{lem-N2}, the total number $A_{q-1}$ of the minimum weight codewords in $\bar{\mathcal{E}_1}$ is equal to the total number of weight $3$ in $\bar{\mathcal{E}_1}^\perp$. Hence $A_{q-1}=(q-1)^2$. Then by Lemma \ref{lem-N1}, the weight enumerator of $\mathcal{E}_1$ is obtained.
\end{proof}

Below we give other constructions of NMDS codes with the same parameters $[q+2, 3, q-1]$ but different weight enumerators.

Define
\begin{eqnarray*}
G_{4,1}=\left[
\begin{array}{lllllll}
1 & 1 & \cdots & 1 & 1 & 1 & 0\\
\alpha_1 & \alpha_2 & \cdots & \alpha_{q-1} & 0 & 1 & 1\\
\alpha_1^2 & \alpha_2^2 & \cdots & \alpha_{q-1}^2 & 1 & 0 & 1
\end{array}
\right].
\end{eqnarray*}
Obviously, $G_{4,1}$ is a $3$ by $q+2$ matrix over $\gf_q$. Let  $\mathcal{F}_{1}$  be the linear code over $\gf_q$ with generator matrix $G_{4,1}$.

The parameters and weight enumerator of $\mathcal{F}_{1}$ are given in the following theorem.
\begin{theorem}\label{thm-N14}
Let $m \geq 3$ be odd. Then the linear code $\mathcal{F}_{1}$ is an NMDS code over $\gf_q$ with parameters $[q+2, 3, q-1]$ and weight enumerator
\begin{eqnarray*}
A(z)=1 + \frac{(q-1)(3q-4)}{2} z^{q-1} + \frac{(q-1)(q^2-6q+14)}{2} z^{q} + \\
\frac{3(q-1)(3q-4)}{2} z^{q+1} + \frac{(q-1)(q-2)^2}{2} z^{q+2}.
\end{eqnarray*}
\end{theorem}
\begin{proof}
The proof is similar to that of Theorem \ref{thm-N2234} and omitted.
\end{proof}

Define
\begin{eqnarray*}
G_{4,2}=\left[
\begin{array}{lllllll}
1 & 1 & \cdots & 1 & 1 & 1 & 1\\
\alpha_1 & \alpha_2 & \cdots & \alpha_{q-1} & 0 & 0 & 1\\
\alpha_1^2 & \alpha_2^2 & \cdots & \alpha_{q-1}^2 & 0 & 1 & 0
\end{array}
\right].
\end{eqnarray*}
 Let  $\mathcal{F}_{2}$ be the linear code  over $\gf_q$ with generator matrix $G_{4,2}$.

In the following theorem, the parameters and weight enumerator of $\mathcal{F}_{2}$ are presented.

\begin{theorem}\label{thm-N15}
Let $m \geq 3$ be odd. Then the linear code $\mathcal{F}_{2}$ is an NMDS code over $\gf_q$ with parameters $[q+2, 3, q-1]$ and weight enumerator
\begin{eqnarray*}
A(z)=1 + (q-1)(q-2) z^{q-1} + \frac{(q-1)(q^2-3q+14)}{2} z^{q} + \\
3(q-1)(q-2) z^{q+1} + \frac{(q-1)(q^2-3q+4)}{2} z^{q+2}.
\end{eqnarray*}
\end{theorem}
\begin{proof}
The proof is similar to that of Theorem \ref{thm-N2234} and omitted.
\end{proof}

Define
\begin{eqnarray*}
G_{4,3}=\left[
\begin{array}{lllllll}
1 & 1 & \cdots & 1 & 1 & 0 & 1\\
\alpha_1 & \alpha_2 & \cdots & \alpha_{q-1} & 0 & 0 & 0\\
\alpha_1^2 & \alpha_2^2 & \cdots & \alpha_{q-1}^2 & 0 & 1 & 1
\end{array}
\right].
\end{eqnarray*}
 Let  $\mathcal{F}_{3}$ be  the linear code over $\gf_q$ with generator matrix  $G_{4,3}$.

The parameters and weight enumerators of $\mathcal{F}_{3}$ are given as follows.
\begin{theorem}\label{thm-N16}
Let $m \geq 3$ be odd. Then the linear code $\mathcal{F}_{3}$ is an NMDS code over $\gf_q$ with parameters $[q+2, 3, q-1]$ and weight enumerator
\begin{eqnarray*}
A(z)=1 + \frac{q(q-1)}{2} z^{q-1} + \frac{(q-1)(q^2+2)}{2} z^{q} + \frac{3q(q-1)}{2} z^{q+1} + \frac{q(q-1)(q-2)}{2} z^{q+2}.
\end{eqnarray*}
\end{theorem}
\begin{proof}
The proof is similar to that of Theorem \ref{thm-N2234} and omitted.
\end{proof}

We remark that the NMDS codes in Theorem \ref{thm-N2234}, \ref{thm-N14}, \ref{thm-N15} and \ref{thm-N16} have different enumerators, though they have the same parameters.
Hence these NMDS codes are pairwise inequivalent to each other. The NMDS code in Theorem \ref{thm-N15}  has the same weight enumerator as that of the NMDS code in \cite[Theorem 12]{Wqiu}. It is open whether they are equivalent to each other.
\section{Optimal locally recoverable codes}\label{sec4}
Locally recoverable codes (LRCs for short) are widely used in distributed data storage systems. In this paper, we only consider linear locally recoverable codes.

For a positive integer $n$, we denote by $[n]=\{ 0,1,\cdots,n-1 \}$.  Let $\mathcal{C}$ be an $[n,k,d]$ linear code over $\mathbb{F}_q$. We index the coordinates of the codewords in $\mathcal{C}$ with the elements in $[n]$.
For each $i \in [n]$, if there exist a subset $R_{i} \subseteq [n] \backslash {i}$ of size $r$ and a function $f_{i}(x_1,x_2,\cdots,x_r)$ on $\mathbb{F}_q^{r}$ meeting $c_i=f_{i}(\mathbf{c}_{R_i})$ for any $\mathbf{c}=(c_0,\cdots,c_{n-1}) \in \mathcal{C}$, then we call $\mathcal{C}$ an $(n,k,d,q;r)$-LRC, where $\mathbf{c}_{R_i}$ is the projection of $\mathbf{c}$ at $R_{i}$. The set $R_{i}$ is called  the repair set of $c_i$.

There exist some tradeoffs between the locality, length, dimension and minimal distance of LRCs.
In the following, two famous bounds on LRCs are presented.

\begin{lemma}[\cite{Gopalan}, Singleton-like bound]\label{lem-R1112}
For any $(n,k,d,q;r)-LRC$,
\begin{eqnarray}\label{eqn-slbound}
d \leq n-k- \left \lceil \frac{k}{r} \right \rceil +2
\end{eqnarray}
\end{lemma}

LRCs are said to be distance-optimal ($d$-optimal for short) when they achieve  the Singleton-like bound.

\begin{lemma}[\cite{Upper}, Cadambe-Mazumdar bound]\label{lem-R1113}
For any $(n,k,d,q;r)-LRC$,
\begin{eqnarray}\label{eqn-cmbound}
k \leq \mathop{\min}_{t \in Z^{+}} [rt+k_{opt}^{(q)}(n-t(r+1),d)]
\end{eqnarray}
where $k_{opt}^{(q)}(n,d)$ is the largest possible dimension of a linear code with length $n$, minimum distance $d$ and alphabet size $q$. $Z^{+}$ represents the set of all positive integers.
\end{lemma}

LRCs are said to be dimension-optimal ($k$-optimal for short) when they achieve  the Cadambe-Mazumdar bound.

\begin{lemma}[\cite{TanP}]\label{lem-R1114}
Let $\mathcal{C}$ be a nontrivial linear code of length $n$, $d^{\perp}=d(\mathcal{C}^{\perp})$. Then $\mathcal{C}$ has the minimum linear locality $d^{\perp}-1$ if and only if
$$
\mathop{\bigcup}_{\mathcal{S} \in \mathcal{B}_{d^{\perp}}(\mathcal{C}^{\perp})} \mathcal{S}=[n].
$$
where $\mathcal{B}_{d}(\mathcal{C})$ denotes the set of the supports of all codewords with weight $d$ in  $\mathcal{C}$, and the coordinates of the codewords are indexed by $(0,1,\cdots,n-1)$.
\end{lemma}

\begin{lemma}[\cite{TanP}]\label{lem-R1115}
Let $\mathcal{C}$ be a nontrivial NMDS code, then the minimum linear locality of $\mathcal{C}$ is either $d(\mathcal{C}^{\perp})-1$ or $d(\mathcal{C}^{\perp})$.
\end{lemma}

\begin{theorem}\label{lem-R1116}
Let $\mathcal{C}$ be an NMDS code, $d^{\perp}=d(\mathcal{C}^{\perp})$. Then the minimum linear locality of $\mathcal{C}^{\perp}$ is $d(\mathcal{C})-1$ if and only if
$$
\mathop{\bigcap}_{\mathcal{S} \in \mathcal{B}_{d^{\perp}}(\mathcal{C}^{\perp})} \mathcal{S}=\emptyset,
$$
where $\mathcal{B}_{d^{\perp}}(\mathcal{C}^{\perp})$ denotes the set of the supports of all codewords with weight $d^{\perp}$ in $\mathcal{C}^{\perp}$.
\end{theorem}

\begin{proof}
The sufficiency was proved in \cite{TanP}. In the following, we  prove the necessity.

Let $\mathcal{C}$ be an $[n,k,d]$ NMDS code. If $\mathcal{C}^{\perp}$ has the minimum linear locality $d(\mathcal{C})-1$, then $\mathop{\bigcup}_{\mathcal{S} \in \mathcal{B}_{d}(\mathcal{C})} \mathcal{S}=[n]$ by Lemma \ref{lem-R1114}.

Now we prove $\mathop{\bigcap}_{\mathcal{S} \in \mathcal{B}_{d^{\perp}}(\mathcal{C}^{\perp})} \mathcal{S}=\emptyset$ in the following. If there exists an integer $i \in [n]$ such that $\mathop{\bigcap}_{\mathcal{S} \in \mathcal{B}_{d^{\perp}}(\mathcal{C}^{\perp})} \mathcal{S}=i$. By Lemma \ref{lem-N2}, it is easy to deduce that $i$ is not in $\mathop{\bigcup}_{\mathcal{S} \in \mathcal{B}_{d}(\mathcal{C})} \mathcal{S}$, which means that $\mathop{\bigcup}_{\mathcal{S} \in \mathcal{B}_{d}(\mathcal{C})} \mathcal{S} \neq [n]$. This contradicts with $\mathop{\bigcup}_{\mathcal{S} \in \mathcal{B}_{d}(\mathcal{C})} \mathcal{S}=[n]$. Then the desired conclusion follows.
\end{proof}

\begin{theorem}\label{thm-R1222}
The NMDS code $\mathcal{C}$ in Theorem \ref{thm-N4235} is a
$$
(q+4,3,q+1,q;2)-LRC
$$
and $\mathcal{C}^{\perp}$ is a
$$
(q+4,q+1,3,q;q)-LRC.
$$
In addition, $\mathcal{C}$ and $\mathcal{C}^{\perp}$ are both d-optimal and k-optimal.
\end{theorem}

\begin{proof}
From the proof of Theorem \ref{thm-N4235}, the support sets of all codewords with weight $3$ in $\mathcal{C}^{\perp}$ are traversed $[q+4]$, which means that
$$
\mathop{\bigcup}_{\mathcal{S} \in \mathcal{B}_{3}(\mathcal{C}^{\perp})} \mathcal{S}=[q+4].
$$
Then by Lemma \ref{lem-R1114}, we have the minimum linear locality of $\mathcal{C}$ is $d(\mathcal{C}^{\perp})-1=2$. Besides, it can be found that the intersection of the support sets of all codewords with weight $3$ in $\mathcal{C}^{\perp}$ is an empty set, which means that
$$
\mathop{\bigcap}_{\mathcal{S} \in \mathcal{B}_{3}(\mathcal{C}^{\perp})} \mathcal{S}=\emptyset
$$
Then by Theorem \ref{lem-R1116}, we have the minimum linear locality of $\mathcal{C}^{\perp}$ is $d(\mathcal{C})-1=q$.
We then prove $\mathcal{C}$ is an optimal LRC. By Lemma \ref{lem-R1112}, putting the parameters of the $(q+4,3,q+1,q;2)$-LRC into the right-hand side of the Singleton-like bound in (\ref{eqn-slbound}), we have
$$
n-k- \left \lceil \frac{k}{r} \right \rceil +2
=q+4-3- \left \lceil \frac{3}{2} \right \rceil +2=q+1.
$$
Hence $\mathcal{C}$ is a $d$-optimal LRC.
By Lemma \ref{lem-R1113}, putting $t=1$ and the parameters of the $(q+4,3,q+1,q;2)$-LRC into the right-hand side of the Cadambe-Mazumdar bound in (\ref{eqn-cmbound}), we have
$$
k \leq r+k_{opt}^{(q)}(n-(r+1),d)
=2+k_{opt}^{(q)}(q+1,q+1) = 3,
$$
where the last equality holds as $k_{opt}^{(q)}(q+1,q+1) = 1$ by the classical Singleton bound. Thus, $\mathcal{C}$ is a $k$-optimal LRC. Then we have proved that $\mathcal{C}$ is both $d$-optimal and $k$-optimal. Similarly, we can prove that $\mathcal{C}^{\perp}$ is both $d$-optimal and $k$-optimal.
\end{proof}

\begin{theorem}
The NMDS code $\mathcal{C}_{1}$ in Theorem \ref{thm-N11} is a
$$
(q+4,3,q+1,q;2)-LRC
$$
and $\mathcal{C}_{1}^{\perp}$ is a
$$
(q+4,q+1,3,q;q)-LRC.
$$
In addition, $\mathcal{C}_{1}$ and $\mathcal{C}_{1}^{\perp}$ are both $d$-optimal and $k$-optimal.
\end{theorem}

\begin{proof}
The proof of this theorem is similar  to that of Theorem \ref{thm-R1222} and omitted.
\end{proof}

\begin{theorem}
The NMDS code $\mathcal{D}$ in Theorem \ref{thm-N4234} is a
$$
(q+3,3,q,q;2)-LRC
$$
and $\mathcal{D}^{\perp}$ is a
$$
(q+3,q,3,q;q-1)-LRC.
$$
In addition, $\mathcal{D}$ and $\mathcal{D}^{\perp}$ are both $d$-optimal and $k$-optimal.
\end{theorem}

\begin{proof}
By the proof of Theorem \ref{thm-N4234}, we can deduce that
$$
\mathop{\bigcup}_{\mathcal{S} \in \mathcal{B}_{3}(\mathcal{D}^{\perp})} \mathcal{S}=[q+3].
$$
and
$$
\mathop{\bigcap}_{\mathcal{S} \in \mathcal{B}_{3}(\mathcal{D}^{\perp})} \mathcal{S}=\emptyset.
$$
The rest of the proof is similar to that of Theorem \ref{thm-R1222} and omitted.
\end{proof}

\begin{theorem}
The NMDS code $\mathcal{D}_{1}$ in Theorem \ref{thm-N12} is a
$$
(q+3,3,q,q;2)-LRC
$$
and $\mathcal{D}_{1}^{\perp}$ is a
$$
(q+3,q,3,q;q)-LRC.
$$
In addition, $\mathcal{D}_{1}$ is both $d$-optimal and $k$-optimal and $\mathcal{D}_{1}^{\perp}$ is $k$-optimal and almost $d$-optimal.
\end{theorem}

\begin{proof}
Similarly to the proof of Theorem \ref{thm-N4234}, we can prove that the intersection of the support sets of all codewords with weight $3$ in $\mathcal{D}_{1}^{\perp}$ is not an empty set. Besides, $\mathop{\bigcup}_{\mathcal{S} \in \mathcal{B}_{3}(\mathcal{D}_1^{\perp})} \mathcal{S}=[q+3].$ Combining  Lemma \ref{lem-R1115} and Theorem \ref{lem-R1116}, we have the minimum linear locality of $\mathcal{D}_{1}^{\perp}$ is $d(\mathcal{D}_{1})=q$, and the minimum linear locality of $\mathcal{D}_{1}$ is $d(\mathcal{D}_{1}^{\perp})-1=2$. The rest of the proof is similar to that of Theorem \ref{thm-R1222}.
\end{proof}

\begin{theorem}
The NMDS code $\mathcal{D}_{2}$ in Theorem \ref{thm-N13} is a
$$
(q+3,3,q,q;2)-LRC
$$
and $\mathcal{D}_{2}^{\perp}$ is a
$$
(q+3,q,3,q;q-1)-LRC.
$$
In addition, $\mathcal{D}_{2}$ and $\mathcal{D}_{2}^{\perp}$ are both $d$-optimal and $k$-optimal.
\end{theorem}

\begin{proof}
Similarly to the proof of Theorem \ref{thm-N4234}, we can prove that
$$
\mathop{\bigcup}_{\mathcal{S} \in \mathcal{B}_{3}(\mathcal{D}_2^{\perp})} \mathcal{S}=[q+3].
$$
and
$$
\mathop{\bigcap}_{\mathcal{S} \in \mathcal{B}_{3}(\mathcal{D}_2^{\perp})} \mathcal{S}=\emptyset.
$$
The remainder  proof of this theorem is similar to that of Theorem \ref{thm-R1222}.
\end{proof}

\begin{theorem}
The NMDS code $\mathcal{E}$ in Theorem \ref{thm-N1235} is a
$$
(q+1,3,q-2,q;2)-LRC
$$
and $\mathcal{E}^{\perp}$ is a
$$
(q+1,q-2,3,q;q-3)-LRC.
$$
In addition, $\mathcal{E}$, $\mathcal{E}^{\perp}$ is both $d$-optimal and $k$-optimal.
\end{theorem}

\begin{proof}
By the proof of Theorem \ref{thm-N1235}, we can prove that
$$
\mathop{\bigcup}_{\mathcal{S} \in \mathcal{B}_{3}(\mathcal{E}^{\perp})} \mathcal{S}=[q+1].
$$
and
$$
\mathop{\bigcap}_{\mathcal{S} \in \mathcal{B}_{3}(\mathcal{E}^{\perp})} \mathcal{S}=\emptyset.
$$
The remainder of the proof is similar to that of Theorem \ref{thm-R1222}.
\end{proof}

\begin{theorem}
The NMDS code $\mathcal{E}_1$ in Theorem \ref{thm-N1234} is a
$$
(q+1,3,q-2,q;3)-LRC
$$
and $\mathcal{E}_1^{\perp}$ is a
$$
(q+1,q-2,3,q;q-3)-LRC.
$$
In addition, $\mathcal{E}_1$ is $k$-optimal and almost $d$-optimal, $\mathcal{E}_1^{\perp}$ is both $d$-optimal and $k$-optimal.
\end{theorem}

\begin{proof}
Similarly to the proof of Theorem \ref{thm-N1235}, we can prove that the support sets of all codewords with weight $3$ in $\mathcal{E}_1^{\perp}$ are not traversed $[q+1]$ because the first and $q$-th coordinates of all codewords with weight $3$ is zero. Besides, $\mathop{\bigcap}_{\mathcal{S} \in \mathcal{B}_{3}(\mathcal{E}_1^{\perp})} \mathcal{S}=\emptyset.$ Combining Lemma \ref{lem-R1115} and Lemma \ref{lem-R1114}, we have the minimum linear locality of $\mathcal{E}_1$ is $d(\mathcal{E}_1^{\perp})=3$, and the minimum linear locality of $\mathcal{E}_1^{\perp}$ is $d(\mathcal{E}_1)-1=q-3$. The rest of the proof is similar to that of Theorem \ref{thm-R1222}.
\end{proof}

\begin{theorem}
The NMDS code $\mathcal{E}_2$ in Theorem \ref{thm-N17} is a
$$
(q+1,3,q-2,q;3)-LRC
$$
and $\mathcal{E}_2^{\perp}$ is a
$$
(q+1,q-2,3,q;q-3)-LRC.
$$
In addition, $\mathcal{E}_2$ is $k$-optimal and almost $d$-optimal, $\mathcal{E}_2^{\perp}$ is both $d$-optimal and $k$-optimal.
\end{theorem}

\begin{proof}
Similarly to the proof of Theorem \ref{thm-N1235}, we can prove that the support sets of all codewords with weight $3$ in $\mathcal{E}_2^{\perp}$ are not traversed $[q+1]$ because the first and $q$-th coordinates of all codewords with weight $3$ is zero. Besides, $\mathop{\bigcap}_{\mathcal{S} \in \mathcal{B}_{3}(\mathcal{E}_2^{\perp})} \mathcal{S}=\emptyset.$ Combining Lemma \ref{lem-R1115} and Lemma \ref{lem-R1114}, we have the minimum linear locality of $\mathcal{E}_2$ is $d(\mathcal{E}_2^{\perp})=3$, and the minimum linear locality of $\mathcal{E}_2^{\perp}$ is $d(\mathcal{E}_2)-1=q-3$. The rest of the proof is similar to that of Theorem \ref{thm-R1222}.
\end{proof}

\begin{theorem}
The NMDS code $\bar{\mathcal{E}_1}$ in Theorem \ref{thm-N2234} is a
$$
(q+2,3,q-1,q;2)-LRC
$$
and $\bar{\mathcal{E}_1}^{\perp}$ is a
$$
(q+2,q-1,3,q;q-2)-LRC.
$$
In addition, $\bar{\mathcal{E}_1}$, $\bar{\mathcal{E}_1}^{\perp}$ are both $d$-optimal and $k$-optimal.
\end{theorem}

\begin{proof}
By the proof of Theorem \ref{thm-N2234}, we deduce that
$$
\mathop{\bigcup}_{\mathcal{S} \in \mathcal{B}_{3}(\bar{\mathcal{E}}^{\perp})} \mathcal{S}=[q+2].
$$
and
$$
\mathop{\bigcap}_{\mathcal{S} \in \mathcal{B}_{3}(\bar{\mathcal{E}}^{\perp})} \mathcal{S}=\emptyset.
$$
The remainder of the proof is similar to that of Theorem \ref{thm-R1222}.
\end{proof}

\begin{theorem}
The NMDS code $\mathcal{F}_{1}$ in Theorem \ref{thm-N14} is a
$$
(q+2,3,q-1,q;3)-LRC
$$
and $\mathcal{F}_{1}^{\perp}$ is a
$$
(q+2,q-1,3,q;q-2)-LRC.
$$
In addition, $\mathcal{F}_{1}$ is k-optimal and almost d-optimal, $\mathcal{F}_{1}^{\perp}$ is both $d$-optimal and $k$-optimal.
\end{theorem}

\begin{proof}
Similarly to the proof of Theorem \ref{thm-N2234}, we can prove that
$$
\mathop{\bigcup}_{\mathcal{S} \in \mathcal{B}_{3}(\mathcal{F}_{1}^{\perp})} \mathcal{S}=[q+2] \backslash \{0\}
$$
and
$$
\mathop{\bigcap}_{\mathcal{S} \in \mathcal{B}_{3}(\mathcal{F}_{1}^{\perp})} \mathcal{S}=\emptyset.
$$
Combining  Lemma \ref{lem-R1114} and \ref{lem-R1115}, we have the minimum linear locality of $\mathcal{F}_{1}$ is $d(\mathcal{F}_{1}^{\perp})=3$,
and the minimum linear locality of $\mathcal{F}_{1}^{\perp}$ is $d(\mathcal{F}_{1})-1=q-2$.
The rest of the proof is similar to that of Theorem \ref{thm-R1222}.
\end{proof}

\begin{theorem}
The NMDS code $\mathcal{F}_{2}$ in Theorem \ref{thm-N15} is a
$$
(q+2,3,q-1,q;3)-LRC
$$
and $\mathcal{F}_{2}^{\perp}$ is a
$$
(q+2,q-1,3,q;q-2)-LRC.
$$
In addition, $\mathcal{F}_{2}$ is $k$-optimal and almost d-optimal, $\mathcal{F}_{2}^{\perp}$ is both $d$-optimal and $k$-optimal.
\end{theorem}

\begin{proof}
Similarly to the proof of Theorem \ref{thm-N2234}, we can prove that
$$
\mathop{\bigcup}_{\mathcal{S} \in \mathcal{B}_{3}(\mathcal{F}_{1}^{\perp})} \mathcal{S}=[q+2] \backslash \{0, q-1\}
$$
and
$$
\mathop{\bigcap}_{\mathcal{S} \in \mathcal{B}_{3}(\mathcal{F}_{1}^{\perp})} \mathcal{S}=\emptyset.
$$
Combining  Lemma \ref{lem-R1114} and \ref{lem-R1115}, we have the minimum linear locality of $\mathcal{F}_{2}$ is $d(\mathcal{F}_{2}^{\perp})=3$,
and the minimum linear locality of $\mathcal{F}_{2}^{\perp}$ is $d(\mathcal{F}_{2})-1=q-2$.
The remainder of the proof is similar to that of Theorem \ref{thm-R1222}.
\end{proof}

\begin{theorem}\label{thm-R1333}
The NMDS code $\mathcal{F}_{3}$ in Theorem \ref{thm-N16} is a
$$
(q+2,3,q-1,q;3)-LRC
$$
and $\mathcal{F}_{3}^{\perp}$ is a
$$
(q+2,q-1,3,q;q-1)-LRC.
$$
In addition, both $\mathcal{F}_{3}$ and $\mathcal{F}_{3}^{\perp}$ are $k$-optimal and almost $d$-optimal.
\end{theorem}

\begin{proof}
Similarly to the proof of Theorem \ref{thm-N2234}, we can prove that
$$
\mathop{\bigcup}_{\mathcal{S} \in \mathcal{B}_{3}(\mathcal{F}_{3}^{\perp})} \mathcal{S}=[q+2] \backslash \{0\}
$$
and
$$
\mathop{\bigcap}_{\mathcal{S} \in \mathcal{B}_{3}(\mathcal{F}_{3}^{\perp})} \mathcal{S}=q+1.
$$
Combining  Lemma \ref{lem-R1114} and \ref{lem-R1115}, we have the minimum linear locality of $\mathcal{F}_{3}$ is $d(\mathcal{F}_{3}^{\perp})=3$,
and the minimum linear locality of $\mathcal{F}_{3}^{\perp}$ is $d(\mathcal{F}_{3})=q-1$.
The remainder of the proof is similar to that of Theorem \ref{thm-R1222}.
\end{proof}

\section{Concluding remarks}

In this paper, based on the oval polynomial $f(x)=x^2$ and some special matrixes, we presented several constructions of NMDS codes.
The weight enumerators of these NMDS codes were explicitly determined.
It is interesting that some constructions produce NMDS codes with the same parameters but different weight enumerators, which comfirms
the fact that NMDS codes with the same parameters may have different weight enumerators. As an important application,
most of these NMDS codes and their duals were proved to be optimal locally recoverable codes.

In \cite{LC}, a class of optimal locally repairable codes of distances 3 and 4 with unbounded length was constructed.
We remark that the optimal locally repairable codes of distance 3 in this paper are not contained in \cite{LC}.


\section*{References}

\begin{thebibliography}{99}
\bibitem{B} M. A. De Boer, Almost MDS codes, Des. Codes and Cryptogr. 9 (2) (1996) 143-155.

\bibitem{Upper}
V. Cadambe, A. Mazumdar, An upper bound on the size of locally recoverable codes, IEEE Int. Symp. Network Coding (2013) 1-5.

\bibitem{D}
C. Ding, Designs from linear codes, World Scientific, Singapore. 2019.

\bibitem{DT} C. Ding, C. Tang, Infinite families of near MDS codes holding $t$-designs, IEEE Trans. Inform. Theory 66 (9) (2020) 5419-5428.

\bibitem{DodLan95}
S. Dodunekov, I. Landgev, On near-MDS codes, J.Geometry 54 (1995) 30-43.

\bibitem{FaldumWillems97}
A. Faldum, W. Willems, Codes of small defect, Des. Codes Cryptogr. 10 (1997) 341-350.

\bibitem{GY}
X. Geng, M. Yang, J. Zhang, Z. Zhou, A class of almost MDS codes, Finite Fields Appl. 79 (2022) 101996.

\bibitem{Gopalan}
 P. Gopalan, C. Huang, H. Simitci, S. Yekhanin, On the locality of codeword symbols, IEEE Trans. Inform. Throry 58 (11) (2012) 6925-6934.

 \bibitem{HDZ}
 Z. Heng, C. Ding, Z. Zhou, Minimal linear codes over finite fields, Finite Fields Appl. 54 (2018) 176-196.

\bibitem{HY}
D. Huang, Q. Yue, Y. Niu, X. Li, MDS or NMDS self-dual codes from twisted generalized Reed-Solomon codes, Designs Codes Cryptogr. 89 (9) (2021) 2195--2209.

\bibitem{Golay49}
M. J. E. Golay, Notes on digital coding, Proc. I.R.E. 37 (1949) 657.

\bibitem{J}
 L. Jin, H. Kan, Self-dual near MDS codes from elliptic curves, IEEE Trans. Inform. Theory 65 (4) (2019) 2166-2170.

\bibitem{L1}
C. Li, Q. Yue, F, Li, Weight distributions of cyclic codes with respect to pairwise coprime order elements, Finite Fields Appl. 28 (2014) 94-114.

\bibitem{L2}
C. Li, P. Wu and F. Liu, On two classes of primitive BCH Codes and some related codes, IEEE Trans. Inform. Theory 65 (6) (2019) 3830-3840.

\bibitem{LN97}
R. Lidl, H. Niederreiter, Finite Fields, Cambridge University Press, Cambridge, 1997.

\bibitem{LC}
Y. Luo, C. Xing, Y. Chen, Optimal locally repairable codes of distance 3 and 4 via cyclic codes, IEEE Trans. Inform. Theory 65 (2) (2018) 1048--1053.

\bibitem{Masch98}
A. Maschietti, Difference sets and hyperovals, Des. Codes Cryptogr. 14(1) (1998) 89-98.

\bibitem{TanP}
P. Tan, C. Fan, C. Ding, Z. Zhou, The minimum linear locality of linear codes, arXiv: 2102. 00597, 2021.

\bibitem{TD}
C. Tang, C. Ding, An infinite family of linear codes supporting $4$-designs, IEEE Trans. Inform. Theory 67 (1) (2020) 244-254.

\bibitem{Wqiu}
Q. Wang, Z. Heng, Near MDS codes from oval polynomials, Discrete Math. 344 (4) (2021) 112277.

\bibitem{W} Y. Wu, Twisted Reed-Solomon codes with one-dimensional hull, IEEE Commun. Letters 25 (2) (2021) 383-386.

\bibitem{WHL}
Y. Wu, J. Y. Hyun, Y. Lee, New LCD MDS codes of non-Reed-Solomon type, IEEE Trans. Inform. Theory 67 (8) (2021) 5069-5078.

		
\end{thebibliography}
\end{document}